\documentclass[aps,pra,showpacs,twoside,twocolumn,10pt]{revtex4-1}
\usepackage[colorlinks=true, citecolor=red, urlcolor=blue ]{hyperref}
\usepackage{epsfig,newlfont,amssymb,amsfonts,amsmath,bm,subfigure,palatino,mathtools,amsthm,braket,times,soul}
\usepackage{enumitem}
\usepackage{cancel}
\usepackage{slashbox}
\DeclareMathOperator{\arctanh}{arctanh}

\usepackage{amsthm}

\theoremstyle{plain}
\newtheorem{thm}{Theorem}

\newtheorem{prop}[thm]{Proposition}

\begin{document}

\title{Response in violation of Bell inequality to imperfect photon addition and subtraction \\ in noisy squeezed states of light}

\author{Saptarshi Roy\(^1\), Titas Chanda\(^1\), Tamoghna Das\(^{1,2}\), Aditi Sen(De)\(^1\), Ujjwal Sen\(^1\)}

\affiliation{\(^1\)Harish-Chandra Research Institute, HBNI, Chhatnag Road, Jhunsi, Allahabad 211 019, India}
\affiliation{\(^2\)Institute of Informatics, National Quantum Information Centre, Faculty of Mathematics, Physics and Informatics, University of Gda\'{n}sk, 80-308 Gda\'{n}sk, Poland.}

\begin{abstract}


Violation of Bell inequality is a prominent detection method for quantum correlations present in composite quantum systems, both in finite and infinite dimensions. We investigate the consequence on the violation of local realism based on pseduospin operators when  photons are added or subtracted in a single mode or in both the modes of the two-mode squeezed states of light in presence of noise. In the noiseless situation,  we show that for addition (subtraction) of photons in a single mode, there is an overall enhancement in the maximal violation, although we observe an interplay between monotonicity and non-monotonicity in the violation of Bell inequality depending on the squeezing strength. Moreover, we report that for low squeezing or low number of photons added or subtracted, subtraction in both the modes can lead to higher violation of local realism than that in the case of addition. For any choice of parameters, such ordering is not seen if one compares their entanglement contents. In the event of a faulty twin-beam generator, we obtain a lower-than-expected squeezing in the state. In such a case, or in imperfect photon addition (subtraction), or under local noise, we find that the violation of local realism by the noise-affected
two-mode squeezed states always decreases. Interestingly however, we notice that photon addition (subtraction) can in general help to conquer the ill-effects of noise by enhancing the violation of local realism or by transforming non-violating states to violating ones, thereby acting as an activating agent.
\end{abstract}

\maketitle

\section{Introduction}
Entangled quantum states \cite{HHHH_RMP} shared between multiple and  distant partners  have the potential of revolutionizing communication and computation schemes \cite{crypto, secret_sharing,  dense_coding, teleport, metro, one_way, sende10,reviewcomm}. 
Historically, the existence of quantum entanglement  was first pointed out in the seminal paper of Einstein-Podolsky-Rosen (EPR) \cite{EPR1935}, which questioned whether the theory of quantum mechanics to be ``incomplete'', based on the assumptions of ``locality'' and ``reality''. John S. Bell formulated a mathematical inequality to be satisfied by any physical theory that is local and realistic, and  which can be violated by entangled quantum states \cite{Bell1964}.   With the development of quantum information science, violation of Bell inequality turned out to be an experimental-friendly detection criterion for  entangled states. Apart from its fundamental importance, violation of Bell inequality has been proven to be
the crucial ingredient in certain proofs of security of quantum cryptography \cite{crypto, secret_sharing}. 
 
Quantum information protocols like entanglement-based quantum key distribution \cite{crypto}, quantum dense coding \cite{dense_coding}, and quantum teleportation \cite{teleport}  were originally proposed for discrete variable systems, and have been implemented, e.g.,  by using the polarization degree of freedom of photons \cite{Pan10photon}.
However, the success probability of preparing  entangled states in this way is very low, and at the same time,  Bell-basis measurement, if required for certain processes is not possible by linear optical elements \cite{Bellbasislinearopt}, thereby making the overall success probability of protocols by using photonic qubits even lower.  It turns out that continuous variable (CV) systems  \cite{CVreview, pati_book} can  overcome certain difficulties, like Bell-basis detection, and hence implementing quantum information processing tasks by using CV states in infinite dimensional systems can be important.
Specifically, they  can be prepared with almost unit probability by using nonlinear interaction of a crystal with laser, and can have only imperfections due to the varying intensity of laser light, resulting in a low squeezing parameter \cite{CVreview}. Therefore, studying quantumness of such CV systems plays a significant role in quantum information science and is the main goal of this article.

Gaussian states, having positive Wigner functions \cite{Wigner},  are one of the most prominent examples of CV states,  advantageous for quantum communication and computation schemes \cite{gaussian_exp_th}.   Although the performance of these states clearly show their nonclassical nature, Bell argued \cite{Bell_epr_epw,epr_epw} that states with positive Wigner functions are naturally endowed with a hidden variable theory, and hence would not violate a Bell inequality. Later, Banaszek and W\'{o}dkiewicz \cite{Banaszek} pointed out that the positivity or negativity of Wigner function has a weak connection to violation of local realism,  and managed to 
 construct a Bell expression out of parity-based operators to obtain violations for two-mode squeezed vacuum (Gaussian) states with positive Wigner functions. However, their technique had intrinsic optimization
  problems \cite{banaszek_fail}, and so even the EPR state,
having maximal quantum correlation,  
   do not violate the inequality maximally. 
In Ref. \cite{Chen},   
   an alternate approach was proposed using pseudospin operators (that are closely related to the parity operators), which is free from such optimization difficulties, and can give the maximal violation in the case of EPR state (for a nice survey see Ref. \cite{nonlocality_positive_wigner}). The pseudospin operators were later generalized \cite{gen_pseudospin, gen_pseudospin2} to
   calibrate the violation of local realism
   for other types of quantum correlated states of continuous variables. Moreover, the violation of Bell inequality for squeezed vacuum states have been tested experimentally using parity-type operators \cite{exp:parity_spatial}, which further motivates the study in this direction.


On the other hand,  there exist several quantum information protocols, like entanglement distillation and quantum error correction \cite{distillable, err_corr_g}, which cannot be performed by Gaussian states with Gaussian operations \cite{Ent_dist_ng, err_corr_ng}. Therefore, over the last few years,  active research has been carried out to investigate properties of non-Gaussian states. One of the simple methods to generate such states is to add or subtract photons \cite{cerf,add_sub_others_ent} to or from the Gaussian states. These processes have also been demonstrated experimentally \cite{add_sub_exp,add_sub_exp2}. Moreover, it was shown that  entanglement content  of the two-mode photon-added (-subtracted) state is much higher than the corresponding two-mode squeezed vacuum states (TMSV) \cite{cerf}, thereby showing enhancement of entanglement due to photon addition or subtraction. 
Moreover, ``degaussification'' via photon addition and subtraction has also been proven to be useful in a variety of situations, like engineering of quantum states to attain hybrid entanglement \cite{hybrid} and for tackling boson sampling problems \cite{boson}.



In this article, we investigate the violations of Bell inequality for photon-added and -subtracted two-mode squeezed vacuum states, both in noiseless and noisy scenarios, where violations of local realism are tested using the pseudospin operators. 
Before considering the imperfection,
we first present the results in the case of single-mode operations without noise, specifically photon addition (subtraction) from a single mode of TMSV, and apart from some instances of diminution, we report an overall enhancement in the maximal violation of Bell inequality with added (subtracted) number of photons. 
However, we report some interesting nonmonotonic features, when odd or even number of photons are added (subtracted) to a single mode.
The response of maximal violation of Bell inequality is also examined, when a given number of photons to be added or subtracted is distributed between the two modes. In particular, in a distributed scenario, we find that unlike entanglement, for a certain squeezing and a small number of added or subtracted photons, subtraction is better than addition according to their quantumness in terms of violation of local realism with pseudospin operators. Moreover, we compare the effect of distribution to single-mode operations, and observe that for sufficiently high squeezing or number of added (subtracted) photons, the maximal violation for distributed operations displays a monotonic enhancement compared to that in case of single-mode operations.

 An important aspect, which turns out to be crucial experimentally, is the 
role of  the inevitable noise that creeps in the TMSV states \cite{noise1} during preparation, transmission, and protocol implementation. 
We investigate the effects of noise on the violation of Bell inequality in two prototypical realistic scenarios -- the states are affected by noise, or when the state generator is itself faulty, i.e., when instead of a TMSV state with certain squeezing, it prepares a state with lower squeezing.  
As expected, noise reduces the amount of quantum correlations present in these states, and hence the amount of violation of local realism. Interestingly, however, we show that the process of photon addition (subtraction)
 can enhance, and in some cases \emph{activate} the violations. 
Specifically, we find that photon addition can transform certain non-violating states to be Bell inequality violating, which we call as the activation of violation of Bell inequality, in a similar spirit that the word ``activation'' was used in the literature for different processes \cite{activation}. In realistic scenarios, even the addition (subtraction) schemes of photons can be faulty due to mechanisms like dark counts \cite{noise_dark,nd2} of the photodetectors. We also study the reaction of the pseudospin operator-based Bell inequality in presence of both noisy and faulty scenarios, and show that activation of violation due to the process of addition (subtraction) of photons is also possible even in presence of two types of noise.

 The paper is organized as follows. In Sec. \ref{sec:formalism}, we discuss about the two-mode squeezed vacuum states, and the effect of photon addition or subtraction on it. The use of pseudospin operators for analyzing violations of Bell inequality is also explored here. In Sec. \ref{sec:single_mode_operations}, the case of single-mode operations is examined. Following it, in Sec. \ref{sec:distribution}, we discuss how distributed operations effect the maximal violation. Furthermore, in Sec. \ref{sec:add_ineq}, we show that in contrary to single-mode operations, in the realm of distributed operations, photon addition is inequivalent to photon subtraction. In Sec. \ref{sec:realistic}, violations of Bell inequality are examined in more realistic scenarios, namely, in the presence of noise in Sec. \ref{sec:noise} and when the squeezed state generator is faulty in Sec. \ref{sec:fault}. Finally, in Sec. \ref{sec:dark}, we deal with the scenario of dark counts in the photon addition and subtraction mechanism, making them erroneous. In Sec. \ref{sec:conclu} contains a conclusion. An appendix provides the proof of the maximization of the Bell expression.



%

%
%

\section{Formalism}
\label{sec:formalism}

Study of Gaussian states lies at the heart of investigations with CV systems. In the state space of Gaussian states, the most general pure states are the displaced squeezed states \cite{CVreview, pati_book}. Since we are interested to study the quantum correlations of quantum states in CV system, and we know that the displacement operator does not alter the  nonlocal properties of a state, in this article without loss of generality, we consider the (undisplaced) two-mode squeezed vacuum (TMSV) state for our investigations. For squeezing strength $r$, the TMSV state can be represented as
\begin{eqnarray}
|\psi_r\rangle = \sum_{n=0}^\infty c_n |n,n\rangle,   
\label{eq:tmsv}
\end{eqnarray}
where $c_n=(1-x)^{\frac{1}{2}} x^{\frac{n}{2}}$ with $x=\tanh^2 r$, and $\lbrace |n\rangle \rbrace$ is the Fock basis consisting of the photon number states. The TMSV state, in the limit of infinite squeezing ($r \rightarrow \infty$), reduces to the well known EPR state.

We can degaussify the TMSV state by simply adding (subtracting) photons locally in its two modes. It was shown  that this degaussification process (photon addition or subtraction) leads to monotonic enhancement of entanglement \cite{cerf}. In this paper, we analyze the effects of photon addition or subtraction on the violation of Bell inequality.
The normalized state after adding $k$ photons to the first mode and $l$ photons to the second mode of the TMSV state reads as
\begin{eqnarray}
|\psi_r^{(k,l)}\rangle = \sum_{n=0}^\infty c_n^{(k,l)}|n+k,n+l\rangle,  
\label{eq:added_state}
\end{eqnarray}
where
\begin{eqnarray}
 c_n^{(k,l)} = \frac{x^{\frac{n}{2}}}{\sqrt{{_2}F_1(k+1,l+1,1,x)}} \sqrt{\binom{n+k}{k}\binom{n+l}{l}}.\nonumber \\
 \label{eq:patmsv}
\end{eqnarray}
 Here, ${_2}F_1$ is the Gauss Hypergeometric function. Note that $c_n^{(0,0)} = c_n$ in Eq. \eqref{eq:tmsv}. On the other hand, the normalized state after subtracting $k$ and $l$ photons from first and second mode respectively is given by
\begin{eqnarray}
|\psi_r^{(-k,-l)}\rangle = \sum_{n=k}^\infty c_n^{(-k,-l)}|n-k,n-l\rangle  
\label{eq:sub_state}
\end{eqnarray}
where
\begin{eqnarray}
 c_n^{(-k,-l)} = \frac{x^{\frac{n-k}{2}}}{\sqrt{{_2}F_1(k+1,k+1,1+k-l,x)}} \sqrt{\frac{\binom{n}{k}\binom{n}{l}}{\binom{k}{l}}}.\nonumber \\
 \label{eq:pstmsv}
\end{eqnarray}
 Without any loss of generality, in this paper, we assume that $k \geq l$. 
 If we restrict operations to a single mode, say, the first mode, the coefficients involved in $|\psi_r^{(\pm k,0)}\rangle$ simplifies as
\begin{eqnarray}
c_n^{(k,0)} = x^{\frac{n}{2}}(1&-&x)^{\frac{1+k}{2}} \sqrt{\binom{n+k}{k}},  
\end{eqnarray}
and
\begin{eqnarray}
  c_n^{(-k,0)} &=& x^{\frac{n-k}{2}}(1-x)^{\frac{1+k}{2}}\sqrt{\binom{n}{k}}.
\label{eq:singlemode_cn's}
\end{eqnarray}


We consider  Bell inequalities  by using the following  pseudospin operators \cite{gen_pseudospin}, given by
\begin{eqnarray}
S^z_q &=& \sum_{\substack{n=0 \\ 2n +q \geq 0}}^\infty |2n +q+1\rangle\langle 2n + q +1| - |2n +q\rangle\langle 2n + q |, \nonumber \\
S^-_q &=& \sum_{\substack{n=0 \\ 2n +q \geq 0}}^\infty |2n +q\rangle\langle 2n + q +1| = (S^+_q)^\dagger,
\label{eq:gen_pseudo_spin}
\end{eqnarray}
where $q$ is an integer. 
The correlation functions for an arbitrary state $\rho$, in terms of the pseudospin operators are given by
\begin{eqnarray}
E(\theta_a,\theta_b) = \text{Tr}[\rho \textbf{S}_{q_1}^{\theta_a}\otimes \textbf{S}_{q_2}^{\theta_b} ], 
\label{eq:correlation_fn}
	\end{eqnarray}
where $\text{S}_{q_i}^{\theta_j} = \cos \theta_j S^z_{q_i} + \sin \theta_j (S^-_{q_i}+S^+_{q_i})$, $j= a, b$,
with $\theta_j$s being the settings of the measurements performed by both the parties, viz. $a$ and $b$.
Like the Clauser-Horne-Shimony-Holt (CHSH) version \cite{chsh_paper} of Bell inequality (Bell-CHSH inequality) in finite dimension, the Bell-CHSH expression in this case based on the  correlation functions, $E(\theta_a,\theta_b)$, in  Eq. \eqref{eq:correlation_fn}, also reads as
\begin{eqnarray}
\chi^{q_1,q_2}_{\theta_a,\theta_b,\theta_a',\theta_b'} = E(\theta_a,\theta_b)+E(\theta_a,\theta_b')+E(\theta_a',\theta_b)-E(\theta_a',\theta_b'). \nonumber \\
\label{eq:bell_expression_gen}
\end{eqnarray}
 Our task is to maximize $\chi^{q_1,q_2}_{\theta_a,\theta_b,\theta_a',\theta_b'}$ (which we refer to as $\chi$ without subscripts and superscripts) with respect to the settings specified by $\theta_a,\theta_b,\theta_a',\theta_b'$, and the pair ($q_1,q_2$). Note that, in the correlation function, constructed out of the pseudospin measurements (see Eq. \eqref{eq:correlation_fn}), we neglect any phase factors since they do not provide any additional information in the maximization of the Bell expression for the states considered here. Therefore, 
we are finally interested to study the properties of a physical quantity, given by 
\begin{eqnarray}
\chi^{max} &=&  \max_{\theta_a,\theta_b,\theta_a',\theta_b',q_1,q_2} \chi^{q_1,q_2}_{\theta_a,\theta_b,\theta_a',\theta_b'} .
\label{eq:max_gen_pseudospin}
\end{eqnarray}
It turns out that the optimisation over the $q$-values can be performed easily by looking at the structure of the concerned state.
Settling with the values of $(q_1,q_2)$, we are left with the optimization over the measurement settings $\lbrace \theta_a,\theta_b,\theta_a',\theta_b' \rbrace$.
 The  correlation function in Eq. \eqref{eq:correlation_fn}, for the states considered in this article, typically, is of the form
\begin{eqnarray}
E(\theta_a,\theta_b) = \pm \cos \theta_a \cos \theta_b + \mathcal{K}\sin \theta_a \sin \theta_b,
\label{eq:correlation_fn_structure}
\end{eqnarray}
 where $ 0 \leq \mathcal{K} \leq 1$. For a TMSV state with squeezing parameter $r$, $\mathcal{K} = \tanh 2r$, and depending on the number of photons added or subtracted to the TMSV state, the $\mathcal{K}_{(\pm k, \pm l)}$ changes accordingly. Here, the subscript of $\mathcal{K}$ denotes the number of photons added or subtracted from each mode of the TMSV state. The optimal measurement settings, which maximizes the violation of pseudospin based Bell inequality considered in Eq. \eqref{eq:bell_expression_gen}, is given by
\begin{eqnarray}
\theta_a = 0,\theta_{a'} = \pi/2,\theta_b = -\theta_{b'} = -\theta,
\end{eqnarray}
where
\begin{eqnarray}
\cos \theta = \frac{1}{\sqrt{1 + \mathcal{K}^2}} \text{ and } \sin \theta = \frac{\mathcal{K}}{\sqrt{1 + \mathcal{K}^2}},
\label{eq:optimal_setting}
\end{eqnarray}
such that the maximal Bell-CHSH quantity, $\chi^{max}$, reduces to
\begin{eqnarray}
\chi^{max} &=& 2\big( \cos \theta + \mathcal{K} \sin \theta \big) = 2\sqrt{1+\mathcal{K}^2}.
\label{eq:max_bv}
\end{eqnarray}
The details of the optimization procedure is given in Appendix \ref{appendix:A} \cite{horo_appendix}. 
 

The amount of enhancement in maximal violation of Bell inequality in the photon addition/subtraction process can be quantified as
\begin{eqnarray}
\mathcal{G}=\frac{\chi^{max}(|\phi\rangle)-\chi^{max}(|\eta\rangle)}{\chi^{max}(|\eta\rangle)}.
\label{eq:gain}
\end{eqnarray}
 Here, we compute the enhancement in the maximal violation of local realism for $|\phi\rangle$ with respect to a given state $|\eta\rangle$. Typically, $\ket{\eta}$ is the TMSV state, while $\ket{\phi}$ is the same TMSV state after adding (subtracting) photons.
 Using the above techniques, we first set out to investigate the effects of maximal violation of Bell inequalities due to single mode operations.

\section{Single mode operations}
\label{sec:single_mode_operations}
Let us first concentrate on the response of maximal violation of Bell inequality of TMSV states subject to addition/subtraction of photons in a single mode.  We first note that, for single mode operations, photon addition in one mode is equivalent to the photon subtraction from the other mode, as it can be easily shown that $|\psi_r^{(k,0)}\rangle = |\psi_r^{(0,-k)}\rangle$ by using Eqs.   
\eqref{eq:added_state}-\eqref{eq:pstmsv}. 
Moreover, since $\chi^{max}(|\psi_r^{(k,0)}\rangle) = \chi^{max}(|\psi_r^{(0,k)}\rangle)$, we can easily see $\chi^{max}(|\psi_r^{(k,0)}\rangle) = \chi^{max}(|\psi_r^{(-k,0)}\rangle)$. 
Therefore, without any loss of generality, we only consider  addition of photons in a single mode, say the first mode, of the TMSV state in this section. However, from an experimental point of view, subtraction is easier to realize than addition \cite{add_sub_exp}, since the latter process essentially requires an additional photon pumping apparatus. So, even if both the processes are equivalent in terms of the maximal violation, experimentally, subtraction is preferred. Furthermore, in situations where addition and subtraction yields inequivalent maximal violation of Bell inequality, it would be noteworthy to find out regions in the relevant parameter space where photon subtraction gives a higher violation than photon addition.
We will address this point in the succeeding section.


\subsection{Addition/subtraction of arbitrary number of photons}
The maximal violation of Bell inequality for the photon-added TMSV state, $|\psi_r^{(k,0)}\rangle$, in the first mode, has the form
\begin{eqnarray}
	\chi ^{max} ( |\psi_r^{(k,0)}\rangle &)& = 2\sqrt{1+ \mathcal{K}_{(k,0)}^2},
	\label{eq:K_single_mode}
\end{eqnarray}
where 
\begin{eqnarray}
	\mathcal{K}_{(k,0)} &=& 2\sum_{n=0}^\infty c_{2n}^{(k,0)}c_{2n+1}^{(k,0)},
	\end{eqnarray}
by using Eqs. \eqref{eq:added_state}-\eqref{eq:patmsv}, and \eqref{eq:max_bv}. The $(q_1,q_2)$-pair which yields this maximal value is $(k \text{ mod } 2,0)$.
	
%
Note that the structure of $\mathcal{K}_{(k,0)}$ remains same for both even and odd number of  photon-added TMSV states. For $k \geq 2$, we obtain the expression for $\mathcal{K}_{(k,0)}$ as
\small
\begin{eqnarray}
&&\mathcal{K}_{(k,0)} = \nonumber \\ &&2(1-x)^{1+k}x^{\frac{1}{2}}\sum_{n=0}^\infty x^{2n} \prod_{i=2}^{k}(2n + i) \sqrt{(2n+1)(2n+k+1)}. \nonumber \\
  \label{eq:maxbv_singlemode}
\end{eqnarray} \normalsize
On the other hand, for  $k=0$ case, i.e., for the  TMSV state, 
$\mathcal{K}_{(0,0)} = 2{x^{1/2}}/{(1+x)} = \tanh 2r$, while if a single photon is added, it takes the form as
\begin{eqnarray}
\mathcal{K}_{(1,0)} = 2(1-x)^{2}x^{\frac{1}{2}}\sum_{n=0}^\infty x^{2n} \sqrt{(2n+1)(2n+2)}.
\label{eq:singlemode_p=0,1}
\end{eqnarray}
Apart from the  TMSV case, $\mathcal{K}_{(k,0)}$ and consequently $\chi^{max}$ cannot be computed analytically for any $k \geq 1$, due to the presence of a square root in the sum involved. Therefore, we resort to approximate methods like series expansion and numerical techniques to compute these summations. For numerical calculations, we first evaluate the above summations upto $n=N$ terms. We then check whether the difference between $\chi^{max}$ with partial sums upto $N$ and $N+1$ terms falls below $10^{-10}$. If this is the case, we conclude that the summation with $N$ terms is sufficient.

Before presenting the results with series expansion, let us discuss the findings with numerical method.
Our analysis reveals an overall enhancement of the maximal violation of Bell inequality in terms of  pseudospin operators of the TMSV state with moderate number of added photons (see Fig. \ref{fig:single_mode1}).  To put the amount of enhancement in a quantitative perspective, we calculate the gain, $\mathcal{G}$, as in  Eq. \eqref{eq:gain}, for some typical values of $r$ and added number of photons, $k$, in the first mode, and is summarized in Table. \ref{table:gain1}. 
As clearly depicted in Fig.  \ref{fig:single_mode1}, there exists a critical value of the squeezing parameter, $r$, beyond which photon addition may lead to a decrement in the maximal violation of Bell inequality, when either a single photon or more than that is added (comparing Figs. \ref{fig:single_mode1}(a) and \ref{fig:single_mode1}(b)). 
Even when we are unable to compute $\chi^{max}$ analytically, the approximate method helps us to obtain the critical value of $r$, $r_c$,  where the violation decreases with addition of photons. In this respect, let us state the following theorem, which shows a special status of a single-photon-addition.

\begin{table}[h]
\begin{center}
\begin{tabular}{ |c| c | c | c | c | } 
\hline
\backslashbox{$r \downarrow$}{$k \rightarrow$} & 2 & 5 & 10 & 15\\
\hline
0.2 & 9.5 & 18.4 & 25.7 & 29.6 \\
\hline
0.5 & 10.0  & 10.9 & 10.7 & 11.4\\
\hline
0.8 & 3.2 & 3.1  & 3.4  & 3.7\\
\hline
1.2 & 0.4 & 0.6 & 0.73 & 0.76 \\
\hline
\end{tabular}
\caption{Percentage of gain, $\mathcal{G}\times 100$, for some typical values of the squeezing parameter, $r$, and added number of photons, $k$, with respect to the TMSV state.}
\label{table:gain1}
\end{center}
\end{table}

\begin{figure}[h]
\includegraphics[width=\linewidth]{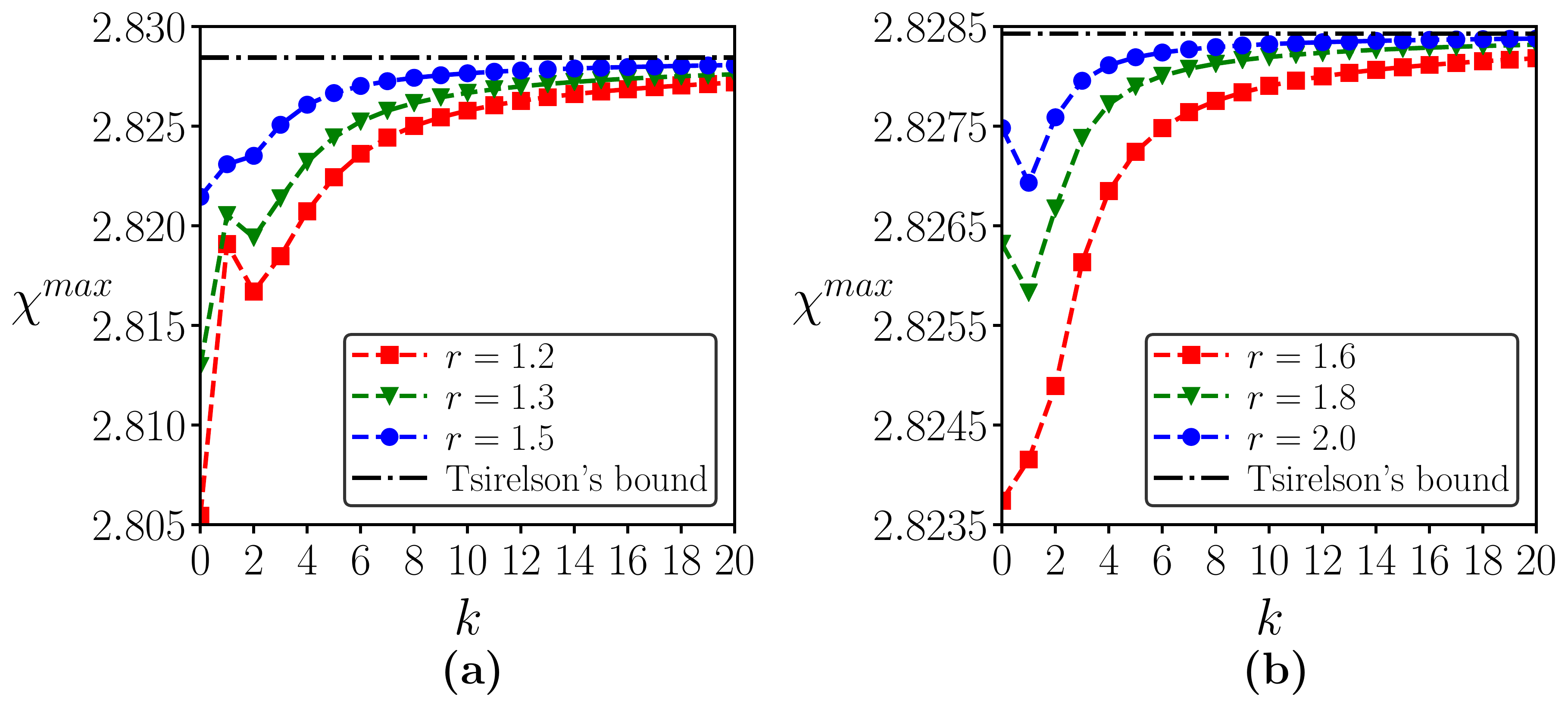}
\caption{Maximal violation of Bell inequality with respect to added (subtracted) number of photons, $k$, in the first mode.
In (a) and (b), different values of the squeezing parameter, $r$, have been considered. In (a), we choose those values of $r$ where $\chi^{max}$ decreases when two or more than two photons are added, while in (b), whenever $\chi^{max}$ shows decreasing nature, it occurs after an addition of a single photon. Clearly, such values of $r$ are above $1.66$ as found in Theorem \ref{th:th1}.
All quantities plotted are dimensionless. 
}
\label{fig:single_mode1}
\end{figure}


 \begin{thm} \label{th:th1}
The maximal violation of Bell inequality based on pseudospin operators shows diminution in comparison to the TMSV state after the addition (subtraction) of a single photon for any finite squeezing parameter, $r$, beyond a critical value, $r_c \approx 1.66$.
\end{thm}

\begin{proof}
We start by approximating the square root term $\sqrt{(2n +1)(2n+2)}$ in  $\mathcal{K}_{(1,0)}$, given in Eq. \eqref{eq:singlemode_p=0,1}. Let $X = 2n+1$ and $Y=2n+2$. Now, using the identity, $(X+Y)^2 - (X-Y)^2 = 4XY$, and putting $(X-Y)^2 = 1$, we get
\begin{eqnarray}
\sqrt{XY} &=& \frac{X+Y}{2}\sqrt{1-\frac{1}{(X+Y)^2}} \nonumber \\
&\approx&  \frac{X+Y}{2}\Big[1 - \frac{1}{2(X+Y)^2}-\frac{1}{8(X+Y)^4}\Big].
\label{eq:am_gm1}
\end{eqnarray}
Substitution of the value of $(X+Y)$ in Eq. \eqref{eq:am_gm1} gives 
\small
\begin{eqnarray}
 \sqrt{(2n+1)(2n+2)} \approx \Big(2n + \frac{3}{2}\Big) - \frac{1}{4(4n + 3)}- \frac{1}{16(4n+3)^3}. \nonumber \\
\label{eq:am_gm_s1}
\end{eqnarray} \normalsize
Under this approximation, $\mathcal{K}_{(1,0)}$  possess a closed form in terms of known standard functions, given by
\small 
\begin{eqnarray}
\mathcal{K}_{(1,0)}^{{approx}} &=& 2\frac{x^{1/2}}{1+x} \Big[\frac{3+x^2}{2(1+x)} -(1-x)^2(1+x)\times \nonumber \\
&\Big(&\frac{{_2}F_1(3/4,1,7/4,x^2)}{12} + \frac{\Phi(x^2,3,3/4)}{2^{10}}\Big)\Big],
\label{eq:am_gm_s1-s0}
\end{eqnarray} \normalsize
where ${_2}F_1$ denotes the Gauss hypergeometric function \cite{arfken},  and $\Phi$ denotes the Lerch transcendent \cite{tran}. Note that the approximation used in Eq. \eqref{eq:am_gm_s1} always leads to an overestimation of $\sqrt{(2n+1)(2n+2)}$, and therefore $\mathcal{K}_{(1,0)}^{{approx}} > \mathcal{K}_{(1,0)}$.
Now the quantity $\mathcal{K}_{(1,0)}^{{approx}} - \mathcal{K}_{(0,0)}$, and consequently $\mathcal{K}_{(1,0)} - \mathcal{K}_{(0,0)}$ and $\chi^{max}(|\psi_r^{(1,0)}\rangle)-\chi^{max}(|\psi_r\rangle)$ becomes negative when $x \gtrsim 0.86$, i.e., $r \gtrsim 1.66$ ($\approx r_c$), and asymptotically approaches to zero from below when $x \rightarrow \infty$. 
Since the process of approximation gives an upper bound of $\mathcal{K}_{(1,0)}$ $\big($or $\chi^{max}(|\psi_r^{(1,0)}\rangle)\big)$, the diminution of maximal violation on adding a single photon persists even without the approximation. Furthermore, keeping upto second order terms is justified, since the next term in the sum near $r = 1.66$, only makes a contributuion of O($10^{-6}$) to the sum. Such observation remains true for all the propositions in this and succeeding sections. Hence the proof. 
\end{proof}
\noindent Note that by numerical simulations, we find the above critical value $r_c$ to be $\approx 1.66$.
Although, in Theorem \ref{th:th1}, we have found the critical value of $r$, beyond which addition of a single photon always leads to diminution of the maximal violation of Bell inequality, increasing the number of added photons results in an overall enhancement of the maximal violation, as mentioned previously. 
However, the enhancement of the violation is accompanied by a seemingly generic nonmonotonic behavior with respect to the squeezing parameter, $r$. We find that there exists a range of the squeezing parameter, $1.42 < r < 1.66$, for which the maximal violation of local realism demonstrates a \emph{monotonic} enhancement with respect to the added number of photons (see Fig. \ref{fig:single_mode1}).
 Apart from the above specified range of the squeezing parameter, $\chi^{max}$ displays a nonmonotonic behavior with the number of added photons (Fig. \ref{fig:single_mode1}). Note however that the nonmonotonicity obtained for $r < 1.42$ with $k$ is different than that of the photon-added state with $r > 1.66$. It is important to stress here that such a feature is absent in the case of entanglement \cite{cerf}.
We now ask whether the critical value of the squeezing parameter for which nonmonotonic to monotonic transition occurs in the behavior of the  maximal violation can be found using the series expansion method.

%
%

We observe from our numerical results, that when the squeezing parameter $r$ is close to the critical value, $1.42$, the transition is dictated by both the values of $\chi^{max}(|\psi_r^{(1,0)}\rangle)$ and $\chi^{max}(|\psi_r^{(2,0)}\rangle)$. 
The value of $\chi^{max}(|\psi_r^{(1,0)}\rangle)$ has been calculated in Eq. \eqref{eq:am_gm_s1-s0} while the evaluation of $\chi^{max}(|\psi_r^{(2,0)}\rangle)$ leads to the following proposition:

\begin{prop}\label{pr:pr1}
 Maximal violation of local realism based on pseudospin operators undergoes a transition from nonmonotonic to monotonic behavior with respect to added number of photons in a single mode for $r \approx 1.42$.
 \end{prop}

\begin{proof}
 By substituting $k=2$ in Eq. \eqref{eq:maxbv_singlemode}, we obtain 
 \small
 \begin{eqnarray}
 \mathcal{K}_{(2,0)} = 2(1-x)^{3}x^{\frac{1}{2}}\sum_{n=0}^\infty x^{2n} (2n + 2) \sqrt{(2n+1)(2n+3)}. \nonumber  \\
 \end{eqnarray}\normalsize
 We approximate $\sqrt{(2n+1)(2n+3)}$ in the same lines as in Eq. \eqref{eq:am_gm1}, and get
 \small
 \begin{eqnarray}
 \sqrt{(2n+1)(2n+3)} \approx 2(n+1) - \frac{1}{4(n+1)} - \frac{1}{8(n+1)^3}.\nonumber \\
 \label{eq:am_gm2}
 \end{eqnarray} \normalsize
 Using the above expression, we get  the approximate value of $\mathcal{K}_{(2,0)}$ as 
 \begin{eqnarray}
\mathcal{K}_{(2,0)}^{approx} &=& \frac{(1-x)^3 x^{1/2}}{2}\Big[ \frac{8x^4}{(1-x^2)^3}-\frac{12x^2}{(1-x^2)^2}\nonumber \\
&-&\frac{7}{2(1-x^2)}-\frac{1}{32x^2}\text{Li}_2(x^2) \Big],
\label{eq:am_gm_s2}
 \end{eqnarray}
 where $\text{Li}_2$ is the polylogarithmic function of order 2. Now, $\mathcal{K}_{(2,0)}^{approx} - \mathcal{K}_{(1,0)}^{approx}$, and consequently $\chi^{max}_{approx}(|\psi_r^{(2,0)}\rangle)-\chi^{max}_{approx}(|\psi_r^{(1,0)}\rangle)$ becomes positive for $x \gtrsim 0.79$, i.e., $r \gtrsim 1.42$. 
 \end{proof}

\noindent \textbf{Remark 1.} The nonmonotonic to monotonic transition and vice-versa in maximal violation of local realism with the added number of photons happens because $\chi^{max}$ first decreases, and then starts increasing before saturating to a certain value with addition.

\noindent \textbf{Remark 2.} The criticalities in $r$ are obtained by keeping upto the second order terms in series (Eqs.\eqref{eq:am_gm_s1} and \eqref{eq:am_gm2}). We find that such approximations nicely match with the values obtained from the numerical simulations.

\noindent \textbf{Remark 3.} The critical squeezing parameters obtained in the above cases can, in principle, be observed in laboratories, as all the critical values of the squeezing parameter are below the maximal amount of experimentally generated squeezing, i.e.,  $r \approx 1.73$ \cite{15db}.

Therefore, by combining the results from Theorem \ref{th:th1} and Proposition \ref{pr:pr1}, we zero in on the squeezing parameter window for which the maximal violation shows monotonic enhancement on adding (subtracting) photons from a single mode using the series expansion method, which clearly agree with Fig. \ref{fig:single_mode1}.
To get more intuitive insights,  we now look at the cases of addition (subtraction) of even or odd number of photons separately in the next subsection.

\subsection{Even-odd dichotomy}

To find out the reason behind such dependence on squeezing parameter of maximal violation, we now study separately the TMSV states when even (odd) number of photons are added. The intuition for such investigation comes from the fact the $\chi^{max}$ depends of $(q_1, q_2)$-pair which is different for odd and even number of photons.
Let us first restrict ourselves to addition of even number of photons from the first mode, and study violations of Bell inequality with respect to the number of photons added for fixed values of the squeezing parameter. From numerical simulations, we find when only even number of photons are added from a particular mode  (Figs. \ref{fig:single_mode_even}(a) and (b)), the maximal violation shows monotonic enhancement for $r \gtrsim 0.94$. We now show the same by using the series expansion method.

\begin{prop} \label{pr:pr2}
The maximal violation of  Bell inequality undergoes a transition from nonmonotonic to monotonic behavior with respect even number of photon addition in a single mode when $r \gtrsim 0.92$.
\end{prop}

\begin{proof}
The insight, obtained from numerical simulation, as reflected in Fig. \ref{fig:single_mode_even}(a), tells us that when $r$ is close to 0.92, the nonmonotonicity can be observed in  the diminution of $\chi^{max}$ after adding atleast $4$ photons to the system. Hence unlike Theorem \ref{th:th1} and Proposition \ref{pr:pr1}, the quantity of interest now becomes $\chi^{max}(|\psi_r^{(4,0)}\rangle)-\chi^{max}(|\psi_r^{(2,0)}\rangle)$. The approximation of $\chi^{max}(|\psi_r^{(2,0)}\rangle)$ has already been done in Eq. \eqref{eq:am_gm_s2}, and hence we are left with the approximation of $\chi_{max}(|\psi_r^{(4,0)}\rangle)$. Expression of $\chi^{max}(|\psi_r^{(4,0)}\rangle)$ is obtained by substituting $k=4$ in Eq. \eqref{eq:maxbv_singlemode}, where
the square root term $\sqrt{(2n+1)(2n+5)}$  can be approximated as
\begin{eqnarray}
\sqrt{(2n+1)(2n+5)} \approx (2n+ 3) - \frac{2}{(2n + 3)} - \frac{2}{(2n + 3)^3}.\nonumber \\ 
\label{eq:am_gm_s4}
\end{eqnarray}
This approximation allows $\mathcal{K}_{(4,0)}$ and thereby  $\chi^{max}(|\psi_r^{(4,0)}\rangle)$ to be written in terms of known functions. Using Eqs. \eqref{eq:am_gm_s2} and \eqref{eq:am_gm_s4}, we find that $\mathcal{K}_{(4,0)}^{approx} - \mathcal{K}_{(4,0)}^{approx}$ and consequently $\chi^{max}_{approx}(|\psi_r^{(4,0)}\rangle)-\chi^{max}_{approx}(|\psi_r^{(2,0)}\rangle) \gtrsim 0$, implies $x \gtrsim 0.51$ or $r \gtrsim 0.92$. 
\end{proof}


\begin{figure}
\includegraphics[width=\linewidth]{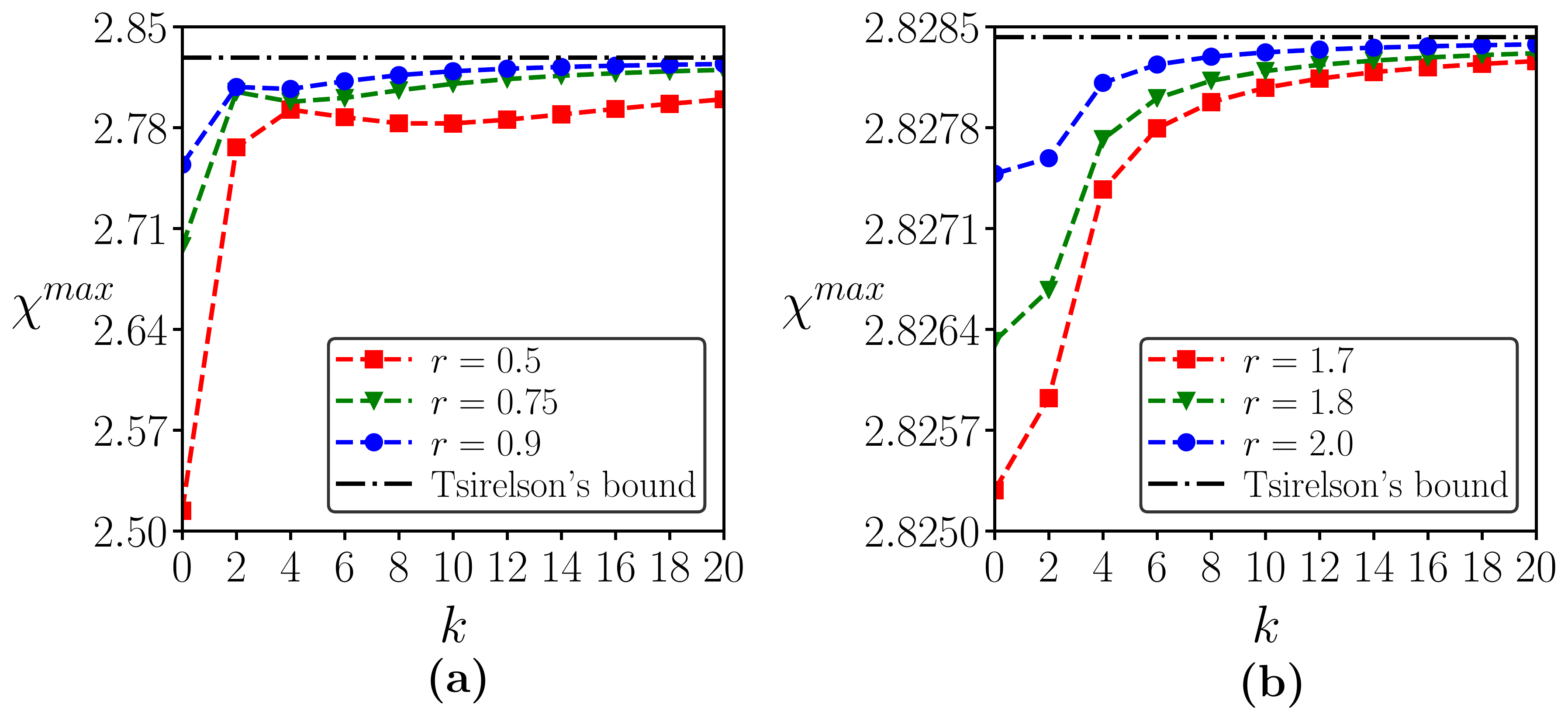} \\
\includegraphics[width=\linewidth]{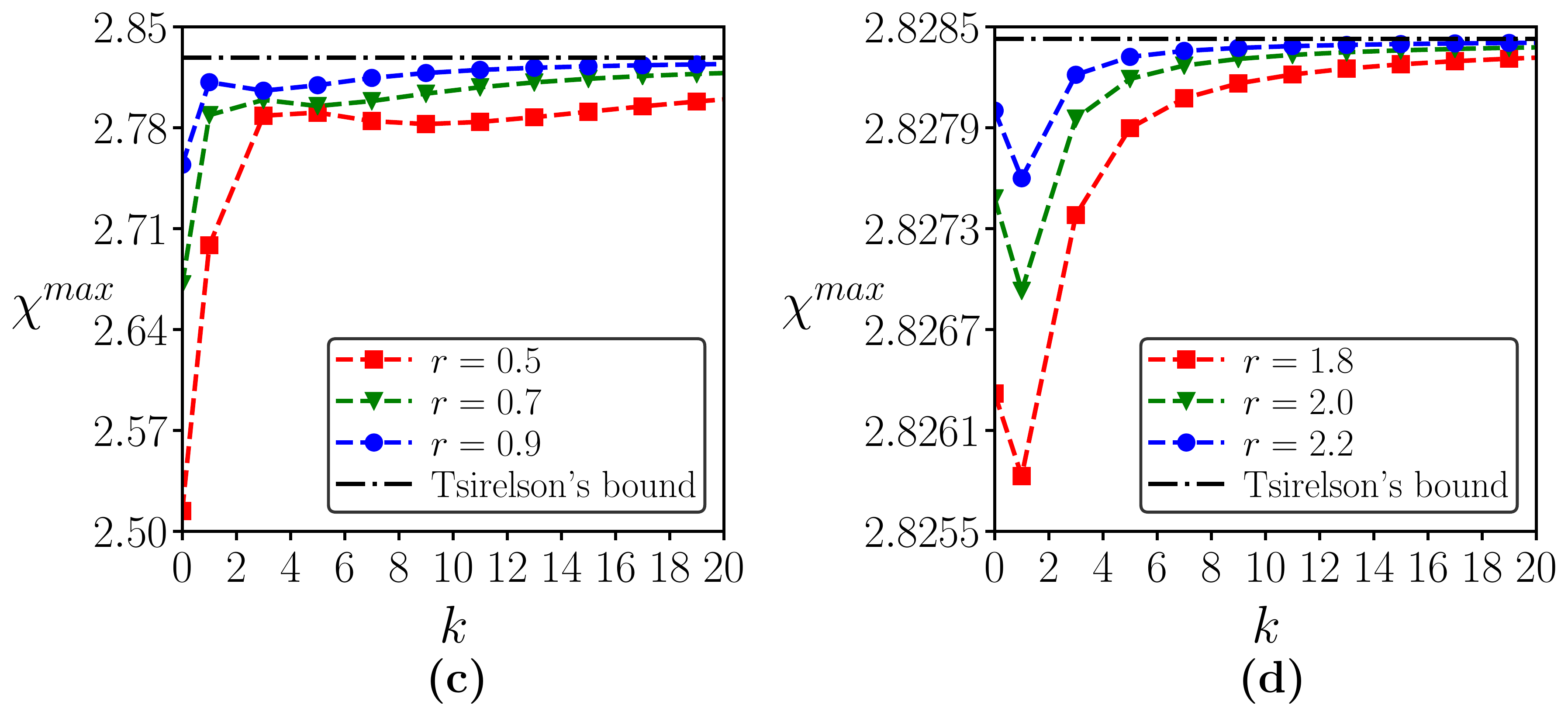}
\caption{Even vs. odd. (Upper panel) $\chi^{max}$ against even number of added photons in a single mode, while (lower panel) $\chi^{max}$ with odd $k$. All quantities plotted are dimensionless.}
\label{fig:single_mode_even}
\end{figure}

We now move to the situation where odd number of photons are added. Interestingly, a qualitatively different picture emerges in this case compared to the even-photon-addition (Figs. \ref{fig:single_mode_even}(c) and (d)). In particular,  there exists only a region in the squeezing parameter, namely $1.23 < r < 1.66$, where we get monotonic behavior of $\chi^{max}$ with odd number of  added photons, $k$. Here, we should also comment that the nature of nonmonotonicity for  $r < 1.23$ and $r > 1.66$ are different, as also seen in Fig. \ref{fig:single_mode1}.
 When $r > 1.66$, the nonmonotonicity is reflected by a decrement in the maximal violation of Bell inequality after addition of a single photon, as already mentioned in Theorem \ref{th:th1}, while for $r < 1.23$, the nonmonotonicity is observed after addition higher number of photons (see Fig. \ref{fig:single_mode_even}(c)). However, the feature of overall enhancement of the violation for higher values of $k$ compared to the TMSV state persists both for even as well as odd $k$. We again employ the series expansion method for obtaining the lower bound on $r$.

\begin{prop} \label{pr:pr3}
 When only odd number of photons are added (subtracted) to a single mode of the TMSV state, the maximal  violation undergoes a transition from nonmonotonic  to monotonic behavior at $r \approx 1.23$.
\end{prop}

\begin{proof}
Again
our numerical results
help us to identify $k$ in $\mathcal{K}_{(k, 0)}$ relevant to prove this proposition. In this case, we notice that  the quantity, $\chi^{max}(|\psi_r^{(3,0)}\rangle)-\chi^{max}(|\psi_r^{(1,0)}\rangle)$ is appropriate. 
The square root term $\sqrt{(2n+1)(2n+4)}$ in $\mathcal{K}_{(3,0)}$ (see Eq. \eqref{eq:maxbv_singlemode}) can be approximately written as
\begin{eqnarray}
\sqrt{(2n+1)(2n+4)} \approx (2n &+& 5/2) - \frac{9}{8(2n + 5/2)}  \nonumber \\ &-& \frac{81}{2^7 (2n + 5/2)^3}.
\label{eq:am_gm_s3}
\end{eqnarray}
This approximation allows $\mathcal{K}_{(3,0)}$ to be written in terms of known hypergeometric and transcendental functions. 
We find $r \gtrsim 1.23$, for which $\chi^{max}_{approx}(|\psi_r^{(3,0)}\rangle)-\chi^{max}_{approx}(|\psi_r^{(1,0)}\rangle) \gtrsim 0$.  
\end{proof}


\noindent \textbf{Remark.} There exists a region, $1.23 < r < 1.42$, where both addition of even and odd number of photons lead to monotonic enhancement of maximal Bell-violation, although the combined curve shows nonmonotonicity. This can be understood by noting the following fact. Even if both even and odd operations give monotonic violation of Bell inequality, it does not guarantee that their combined effect would be   monotonic. Hence, individual monotonicity is necessary but not a sufficient condition for combined monotonicity. However,  nonmonotonic maximal Bell-violation in individual cases ensures the violation for the state with arbitrary number of added photons to be nonmonotonic.

We know that the violation of Bell inequality by quantum states quantifies the content of quantum correlations present in these states. Another way to quantify quantum correlation is the amount of entanglement possessed by these states. Comparing the results obtained here with the entanglement content \cite{cerf} of photon-added (-subtracted) TMSV states, we observe that the monotonic relationship of these quantities for pure  two-qubit states
in finite dimension is no longer true for pure states in the continuous variable case provided the Bell test is performed with pseudospin operators.

\section{Distributed photon addition and subtraction}
\label{sec:distribution}
In this section, we go beyond the realm of single mode operations, and analyze the effect of local operations on both the modes. Specifically, for a given number of photons to be added, instead of dumping them in a single mode, we distribute  them in two modes, and examine the effects of distribution on violations of Bell inequality based on pseudospin operators. Similar operations are considered in case of photon subtraction, which in this case, is different from photon addition. For addition, the correlation function of the state given in Eq. \eqref{eq:patmsv}, in terms of pseudospin operators, also takes the form as in Eq. \eqref{eq:correlation_fn_structure} with
\begin{eqnarray}
	\mathcal{K}_{(k,l)} &=& 2 \times \text{max} \Big[ \sum_{n=0}^\infty c_{2n}^{(k,l)}c_{2n+1}^{(k,l)}, \sum_{n=0}^\infty c_{2n+1}^{(k,l)}c_{2n+2}^{(k,l)}\Big],   \nonumber \\
	\label{eq:K_dis_add}
\end{eqnarray}
where $c_n^{(k,l)}$s are given in Eq. \eqref{eq:patmsv}. Note that the maximization in the above equation arises due to the optimization involved in $(q_1,q_2)$-duo. In the case of photon subtraction, we can rewrite the state given in Eq. \eqref{eq:pstmsv} as $|\psi_r^{(-k,-l)}\rangle = \sum_{n=0}^\infty c_{n+k}^{(-k,-l)}|n,n+k-l\rangle$. For this case,
\begin{eqnarray}
	\mathcal{K}_{(-k,-l)} &=& 2 \times \nonumber \\ \text{max} &\Big[&\sum_{n=0}^\infty c_{2n+k}^{(-k,-l)}c_{2n+1+k}^{(-k,-l)},\sum_{n=0}^\infty c_{2n+1+k}^{(-k,-l)}c_{2n+2+k}^{(-k,-l)} \Big ], \nonumber \\  
	\label{eq:K_dis_sub}
	\end{eqnarray} 
	where $c_n^{(-k,-l)}$s are represented in Eq. \eqref{eq:pstmsv}. The corresponding maximal violation of Bell inequality for both addition and subtraction reads as
	\begin{eqnarray}
	\chi ^{max} ( |\psi_r^{(\pm k,\pm l)}\rangle &)& = 2\sqrt{1+ \mathcal{K}_{(\pm k,\pm l)}^2}.
	\label{eq:K_dis_bv}
\end{eqnarray}
Note that in general, $\mathcal{K}_{(k,l)} \neq \mathcal{K}_{(-k,-l)}$ and hence $\chi^{max}( |\psi_r^{( k, l)}\rangle)$ and $\chi^{max}( |\psi_r^{(-k, -l)}\rangle)$ are typically different. However, when $k=l$, we notice that $ |\psi_r^{( k, k)}\rangle$ and $ |\psi_r^{(-k, -k)}\rangle$ have the same Schmidt coefficients \cite{cerf}, and consequently, $\chi ^{max} ( |\psi_r^{( k, k)}\rangle) = \chi ^{max} ( |\psi_r^{(- k,- k)}\rangle)$. 
Therefore, when $k \neq l$ ($k \neq 0$), there exists a disparity in the maximal violation of Bell inequality for distributed photon-added and -subtracted states (see Figs. \ref{fig:add_sub_ineq} and \ref{fig:add_sub_distri}). We  discuss this inequivalence in this section.

\begin{figure}[h]
\includegraphics[width=\linewidth]{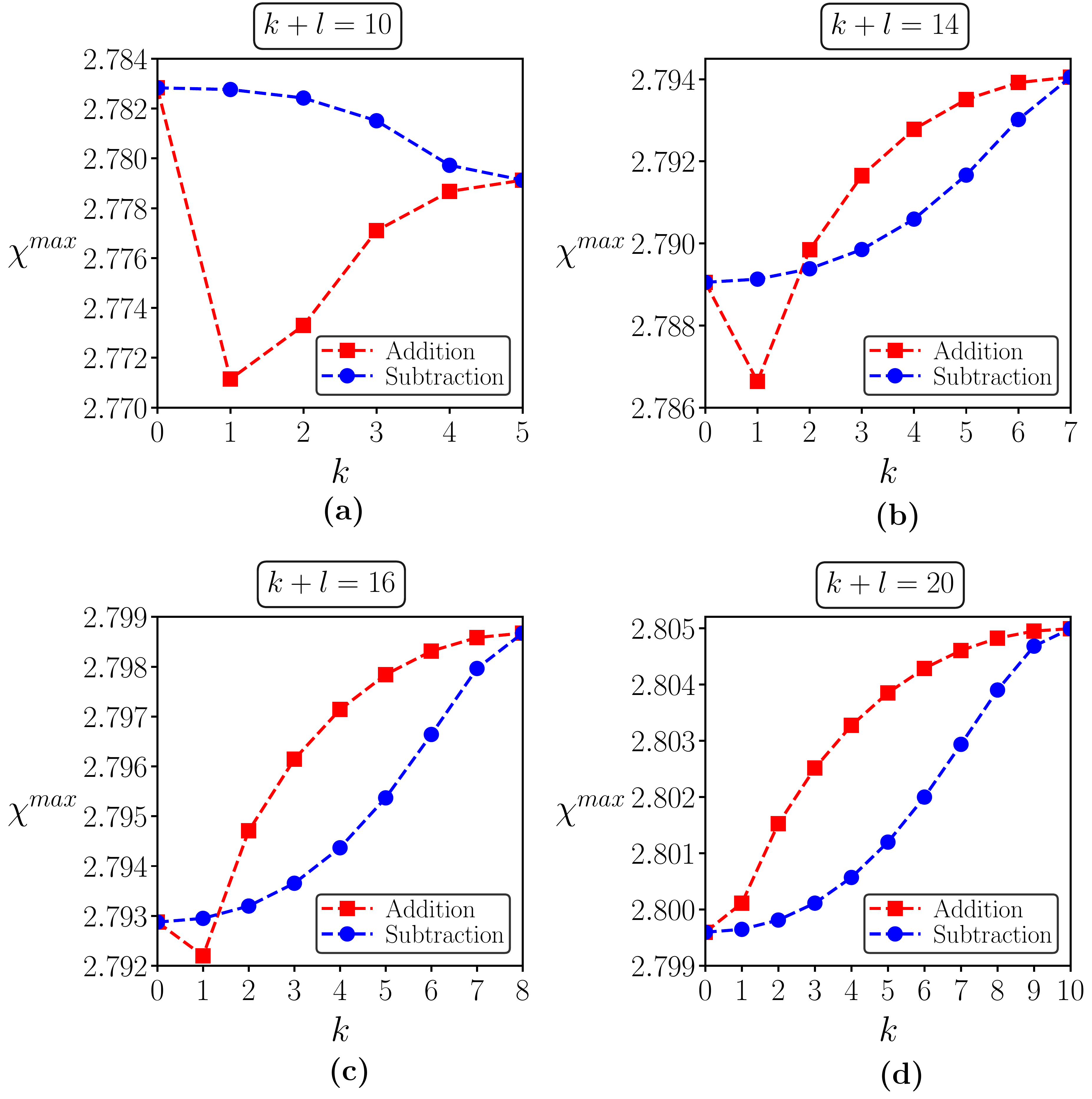}
\caption{Inequivalence of distributed addition and subtraction of photons. Here $r = 0.5$. The abcissa represents the number of photons added or subtracted from the first mode for a  given  total number of photons, $k+l$, where $k$ and $l$ represent the photons added (subtracted) in the first and the second modes respectively. Furthermore, it highlights a relationship between nonmonotonicity and the relative performance of distributed addition (subtraction) in terms of their maximal violation of Bell inequality. Both the axes are dimensionless.}
\label{fig:add_sub_ineq}
\end{figure}

\subsection{Inequivalence of addition and subtraction}
\label{sec:add_ineq}
%
%
%
%
%

In Sec. \ref{sec:single_mode_operations}, we argued that both addition and subtraction of photons from a single mode yields the same maximal violation. 
As pointed out earlier, this equivalence breaks down in case of distributed photon-addition and -subtraction (Figs. \ref{fig:add_sub_ineq} and \ref{fig:add_sub_distri}). This inequivalence prompts us a natural question.-- In terms of maximal violation under distribution, which one is better-- addition or subtraction? 
\begin{figure}[t]
\includegraphics[width=\linewidth]{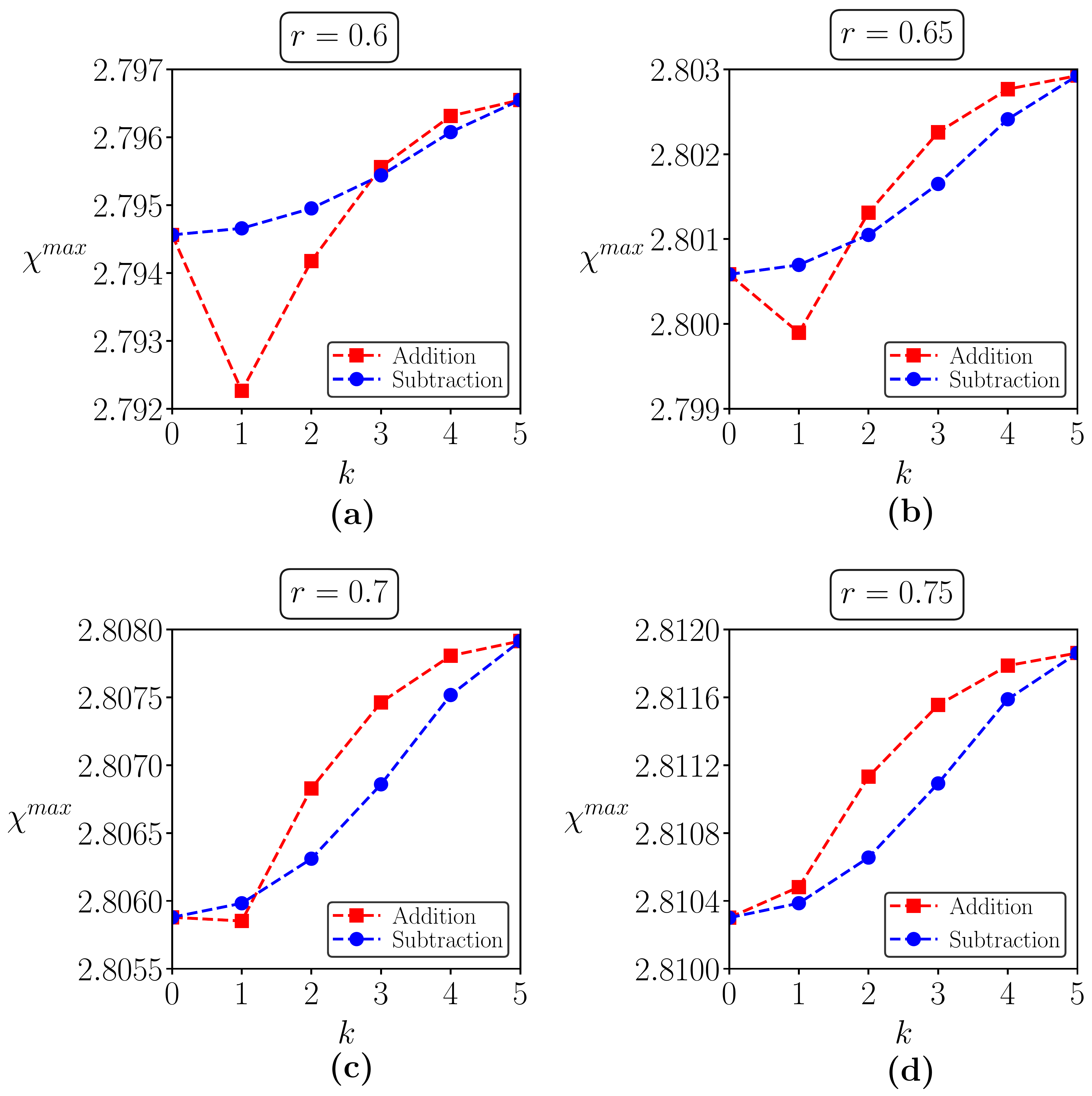}
\caption{Washing away of diminution and nonmonotonicity in the maximal violation with increasing squeezing parameter. The total number of photons added (subtracted) is always fixed to $10$. Here, the abcissa, $k$, is the added number of photons in the first mode. Other details are same as in Fig. \ref{fig:add_sub_ineq}.}
\label{fig:add_sub_distri}
\end{figure}


To answer the above question, for a fixed squeezing pararmeter $r$, and for fixed total number of photons added (subtracted) in both the modes, $k+l$, with $k$ and $l$ being the added (subtracted) photons from the first and second modes respectively, we investigate the behavior of $\chi^{max}$ with respect to $k$.
Extensive numerical analysis reveals the following qualitative trends of $\chi^{max}$ under distributed operations (certain exemplary scenarios are depicted in Figs. \ref{fig:add_sub_ineq} and \ref{fig:add_sub_distri}).

\begin{enumerate}[label=\textbf{\arabic*.}]
\item In case of both addition and subtraction, for `low' values of squeezing parameter $r$ and the total number of added (subtracted) photons, $k+l$, we observe that the maximal violation usually decreases with the number of photons added in the first mode, $k$, which sometimes leads to the nonmonotonicity of $\chi^{max}$ against $k$. Moreover, in the distributed case, the value of $\chi^{max}$ occasionally turns out to be smaller compared to that of the single mode operations.

\item For `low' to `intermediate' values of $r$ and $k+l$, interestingly, we find that 
distributed subtraction gives more violation compared to distributed addition for some specific values of $r$ and $k+l$.

\item The traits of diminution and nonmonotonicity get completely washed away to monotonic enhancement of maximal violation for `sufficiently high' $r$ or by `increasing' the total number of added or subtracted photons $k+l$ or both. In this parameter regime, distributed addition typically yields a higher violation compared to distributed subtraction.

\item For distributed addition, the transition from nonmonotonicity and diminution to monotonic enhancement of maximal violation usually requires `higher' values of squeezing, $r$, or, total number of photons, $k+l$, compared to the distributed subtraction case.
\end{enumerate}


The observations are in sharp contrast to the results obtained in the case of entanglement \cite{cerf}, where distribution always leads to monotonic enhancement of entanglement for both addition and subtraction. Furthermore, distributed addition is shown to ubiquitously outperform distributed subtraction in terms of the entanglement content (cf. \cite{tamo_4mode}). As argued above, this is no more true in case of violation of Bell inequality.
Moreover careful survey in the space of squeezing parameter and total number of photons added or subtracted indicate that the outcome of this duel (addition vs. subtraction) has a one to one correspondence with monotonicity of maximal violation upon distribution of the added or subtracted photons. The general trend being when maximal violation for distributed addition shows nonmonotonicity or diminution, subtraction prevails, which as pointed out earlier occurs for `low' to `intermediate' values of $r$ and $k+l$. 


%
%

%

\section{Violations of Bell inequality in Realistic Situations}
\label{sec:realistic}

The cases considered so far are ideal, as the TMSV states  were not reckoned to be tampered  by any noise due to environmental interactions and the twin beam generator was assumed to be without any imperfections. However, in laboratories, presence of noise and faulty machines are generic  \cite{noise1}. In this section, we address these issues, and focus on imperfect (noisy and faulty) scenarios which reduce the maximal violation of Bell inequality,  and in some cases, even makes the system non-violating. We show how even single mode operations, namely addition or subtraction of photons can enhance  violation of Bell inequality in these scenarios, and sometimes can even \emph{activate} violation for states which ceased to violate Bell inequalities in presence of noise or imperfections. 

Here we consider two major sources of imperfections that can have detrimental effect on the maximal violation of Bell inequality. {\bf 1.} We consider the case of a general local noise model, and examine its effect on the  violation of Bell inequality for TMSV states. We then  analyze
 enhancement and/or activation of the violation via photon addition or subtraction, giving examples for specific cases of local thermal and Gaussian noise. We also repeat the same analysis for a classically correlated noise model. {\bf 2.}  We assume that there is a faulty  faulty twin beam generator, resulting a TMSV state with squeezing, different than the desired one and perform the same investigations like effects on violation on local realism due to states with defects as in the case of noisy states. {\bf 3.} We 
consider the situation where the photon addition and subtraction procedures are themselves faulty due to features like dark counts \cite{noise_dark,nd2} etc. of the photodetectors employed during the photon addition and subtraction procedures.
 


\subsection{Noise in states}
\label{sec:noise}

We now look at the TMSV states, tampered by noise, and study its robustness against such mixing in terms of its ability to violate the Bell inequality based on the pseudospin operators. The violation is computed in two distinct scenarios: \textit{(i)} when the probability with which the noise gets mixed with the TMSV state is known and, \textit{(ii)} when the information about the mixing probability is absent. In the first case, for a given $p$, the maximal violation of Bell inequality is evaluated while in the second one, the settings chosen for  optimizing the violation of Bell inequality is same as the one with vanishing $p$. In both the cases, we analyze the effects of photon addition and subtraction on the violation of Bell inequality.

\subsection*{Local noise}
We consider a general local noise model, where  the noisy state reads as
\begin{eqnarray}
\rho = (1-p)|\psi_r\rangle \langle \psi_r| + p \Big( \sum_{n=0}^{\infty} \mu_n |n \rangle \langle n| \otimes  \sum_{m=0}^{\infty} \nu_m |m \rangle \langle m| \Big), \nonumber \\
\label{eq:noise_local}
\end{eqnarray}
where $0\leq p \leq 1$, $\ket{\psi_r}$ is the TMSV state with squeezing parameter $r$, and
 $\sum_{n=0}^{\infty} \mu_n = \sum_{m=0}^{\infty} \nu_m = 1$.
The correlation function for $\rho$, following Eq. \eqref{eq:correlation_fn}, in terms of the pseudospin operators is given by
\begin{eqnarray}
E(\theta_a,\theta_b) = A\big(\cos \theta_a \cos \theta_b + \frac{B}{A}\sin \theta_a \sin \theta_b \big),
\end{eqnarray}
with
\begin{eqnarray}
A &=& (1-p) + p\Big(\sum_{n=0}^{\infty}(-1)^n \mu_n \Big)\Big( \sum_{m=0}^{\infty} (-1)^m\nu_m \Big),\nonumber \\
B &=& (1-p)\tanh 2r. 
\label{eq:noise_A,B}
\end{eqnarray}
In practical situations, the knowledge of $p$, i.e., whether any error have acted or not, may be elusive. Therefore, two situations may arise: \emph{(i)} when the value of $p$ is known, and \emph{(ii)} when it is unknown.
The maximum value of the Bell expression for the state $\rho$, when the mixing probability, $p$, is known, is given by (see Eqs. \eqref{eq:optimal_setting} and \eqref{eq:max_bv})
\begin{eqnarray}
 \chi^{max}_{p}(\rho) =  2\sqrt{A^2 +B^2}.
 \label{eq:noise_bv_p_known}
\end{eqnarray} 
When the knowledge about $p$ is absent, 
one might proceed with the optimal measurement setup for the TMSV state, $|\psi_r \rangle$,  and calculate the violation.
 Bell expression for such a setting of the state given in Eq. \eqref{eq:noise_local} reads as
\begin{eqnarray}
 \chi^{max}_{\bcancel{p}}(\rho) = 2\Big( \frac{A+\mathcal{K}_{(0,0)}B}{\sqrt{1 + \mathcal{K}^2_{(0,0)}}}\Big) = 2 \Big(\frac{A + B \tanh 2r}{\sqrt{1+\tanh^2 2r}}\Big). \nonumber \\
\label{eq:noise_bv_p_unknown}
\end{eqnarray}
When the value of $p$ is known, we have $\chi^{max}_{p}(\rho)>2$, when
\begin{eqnarray}
p < 1 - \frac{1}{a^2+b^2}\Big( a(a-1) + \sqrt{a(a - ab^2 + 2b^2)} \Big),
\label{eq:noise_p_known_critical}
\end{eqnarray}
where $a = 1 - \Big(\sum_{n=0}^{\infty}(-1)^n \mu_n \Big)\Big( \sum_{m=0}^{\infty} (-1)^m\nu_m \Big)$ and $b = \tanh 2r$. When the knowledge about the value of $p$ is absent, then $\chi^{max}_{\bcancel p}(\rho)>2$, if
\begin{eqnarray}
p <  \frac{\sqrt{1+b^2}(\sqrt{1+b^2}-1)}{2+b^2 - a}. 
\label{eq:noise_p_cri_unknown}
\end{eqnarray}

We now explore the possibilities of enhancement and/or activation of the violation via addition (subtraction) of photons to one of the modes of such  noisy states. We assume, without any loss of generality, that the single mode operations are performed in the first mode. The normalized state when $k$ photons are added in the first mode of $\rho$, given in Eq. \eqref{eq:noise_local}, can be represented as
  \begin{eqnarray}
\tilde{\rho}_k &=& (1-p)|\psi_r^{(k,0)} \rangle \langle \psi_r^{(k,0)} | \nonumber \\
&+& p \Big( \sum_{n=0}^{\infty} \tilde{\mu}_n^{k} |n+k \rangle \langle n+k| \otimes  \sum_{m=0}^{\infty} \nu_m |m \rangle \langle m| \Big). 
\label{eq:noise_local_add_k}
\end{eqnarray}
where 
\begin{eqnarray}
\tilde{\mu}_n^{k} = \frac{\mu_n \binom{n+k}{k}}{\sum_{t=0}^{\infty}\mu_t \binom{t+k}{k}}.
\end{eqnarray}
 When $k$ photons are subtracted from $\rho$, we have 
 \begin{eqnarray}
\tilde{\rho}_{-k} &=& (1-p)|\psi_r^{(-k,0)} \rangle \langle \psi_r^{(-k,0)} | \nonumber \\
&+& p \Big( \sum_{n=0}^{\infty} \tilde{\mu}_n^{-k} |n \rangle \langle n| \otimes  \sum_{m=0}^{\infty} \nu_m |m \rangle \langle m| \Big), \nonumber \\
&=& (1-p)|\psi_r^{(0,k)} \rangle \langle \psi_r^{(0,k)} | \nonumber \\
&+& p \Big( \sum_{n=0}^{\infty} \tilde{\mu}_n^{-k} |n \rangle \langle n| \otimes  \sum_{m=0}^{\infty} \nu_m |m \rangle \langle m| \Big).
\label{eq:noise_local_sub_k}
\end{eqnarray}
with 
\begin{eqnarray}
\tilde{\mu}_n^{-k} = \frac{\mu_{n+k} \binom{n+k}{k}}{\sum_{t=0}^{\infty}\mu_{t+k} \binom{t+k}{k}}.
\end{eqnarray}
Here, the forms of $\ket{\psi^{\pm k, \pm l}}$ are given in Eqs. \eqref{eq:added_state} and \eqref{eq:sub_state}.
 The correlation functions corresponding to states in Eqs. \eqref{eq:noise_local_add_k} and \eqref{eq:noise_local_sub_k} have the same structure as Eq. \eqref{eq:noise_A,B},
and the corresponding maximal Bell inequality violation, when $p$ is known, is given by
\begin{eqnarray}
\chi^{max}_p (\tilde{\rho}_{\pm k})= 2 \sqrt{A_{\pm k}^2 + B_{\pm k}^2}.
\end{eqnarray}
For addition of photons in the first mode, with the optimal $(q_1,q_2)$-pair, $A_{+k}$ and $B_{+k}$ takes the following form
\begin{eqnarray}
A_{+k} &=& (1-p) + p\Big(\sum_{n=0}^{\infty}(-1)^n \tilde{\mu}_n^k \Big)\Big( \sum_{m=0}^{\infty} (-1)^m\nu_m \Big),\nonumber \\
B_{+k} &=& (1-p) \mathcal{K}_{(k,0)}, 
\label{eq:noise_A,B_add_k}
\end{eqnarray}
and in case of photon subtraction,
\begin{eqnarray}
A_{-k} &=& (1-p) + (-1)^k p\Bigg[ \Big(\sum_{n=0}^{\infty}(-1)^n \tilde{\mu}_n^{-k} \Big)  \nonumber \\
&& \Big( \sum_{m=k \text{ mod } 2}^{\infty}  (-1)^m\nu_m \Big) \Bigg], \nonumber \\
 B_{-k} &=& (1-p) \mathcal{K}_{(-k,0)}.
\label{eq:noise_A,B_sub_k}
\end{eqnarray}
When the knowledge about $p$ is absent, the measurement settings which are optimal for the $k$ photon-added TMSV state are employed. The maximal  Bell expression for such a setting for the photon-added and -subtracted noisy state, $\tilde{\rho}_{\pm k}$, is as follows:
\begin{eqnarray}
\chi^{max}_{\bcancel{p}}(\tilde{\rho}_{\pm k}) &=& 2\Big( \frac{A_{\pm k}+\mathcal{K}_{(\pm k,0)}B_{\pm k}}{\sqrt{1 + \mathcal{K}^2_{(\pm k,0)}}}\Big)  \nonumber \\
&=& 2\Big( \frac{A_{\pm k}+\mathcal{K}_{(k,0)}B_{\pm k}}{\sqrt{1 + \mathcal{K}^2_{(k,0)}}}\Big).
\label{eq:noise_bv_p_unknown_add_k}
\end{eqnarray}
Note that unlike in the noiseless scenario, in the presence of local noise, the maximal violation for single mode addition and subtraction are structurally different (see Eqs. \eqref{eq:noise_A,B_add_k} and \eqref{eq:noise_A,B_sub_k}). 
We now consider two special cases where the local noises considered in Eq. \eqref{eq:noise_local} are thermal  and  Gaussian. In both these cases, the system ceases to violate the Bell inequality after a critical value of $p$, even when the value of $p$ is known.


\subsubsection{Local thermal noise}
\label{sec:local_th_noise}

Let us first consider the scenario of local thermal noise. In this situation, the resulting state, $\rho^{\beta_1\beta_2}$, is the admixture of TMSV state with the thermal noise having   inverse temperatures, $\beta_1= \frac{1}{k_B T_1}$ and $\beta_2 = \frac{1}{k_B T_2}$, for first and second modes respectively with $k_B$ being the Boltzman constant and $t_i$, $i=1,2$ being the temperature of the $i^{th}$ mode. The local thermal noise parameters are given by 
\begin{eqnarray}
\mu_n = (1 - e^{-\beta_1})e^{-\beta_1 n}, \nonumber \\  \nu_m = (1 - e^{-\beta_2})e^{-\beta_2 m}.
\end{eqnarray}
 For these choices of noise parameters, we obtain 
\begin{eqnarray}
A=(1-p) + p\tanh \frac{\beta_1}{2} \tanh \frac{\beta_2}{2},
\end{eqnarray} 
  while $B$ remains the same as in Eq. \eqref{eq:noise_A,B}. The maximal violation, when the value of $p$ is known, reduces to (see Eq. \eqref{eq:noise_bv_p_known})
  \small
\begin{eqnarray}
&&\chi_{p}^{max}(\rho^{\beta_1 \beta_2}) = 2 \times \nonumber \\ 
&& \sqrt{(1-p)^2 \tanh^2 2r + \Big\lbrace(1-p) + p\tanh \frac{\beta_1}{2} \tanh \frac{\beta_2}{2}\Big\rbrace^2}. \nonumber \\
\label{eq:noise_thermal_BV_gen}
\end{eqnarray} \normalsize
The range of mixing probability, $p$, for which $\rho^{\beta_1 \beta_2}$ violates Bell inequality, is given in Eq. \eqref{eq:noise_p_known_critical},
with $a$ and $b$ are now as follows:
\begin{eqnarray}
a &=& 1 - \tanh \frac{\beta_1}{2} \tanh \frac{\beta_2}{2}, \nonumber \\ b &=& \tanh 2r. 
\label{eq:noise_p_known_critical_thermal_a_b}
\end{eqnarray}
When the value of $p$ is unknown, following Eq. \eqref{eq:noise_bv_p_unknown}, the violation is given by
\begin{eqnarray}
\chi_{\bcancel{p}}^{max}(\rho^{\beta_1 \beta_2}) &=& \frac{2}{\sqrt{1+\tanh^2 2r}}\Big[(1-p) (1+\tanh^2 2r) \nonumber \\ 
&+& p  \tanh \frac{\beta_1}{2} \tanh \frac{\beta_2}{2}\Big].
\label{eq:noise_thermal_BV_p_unknown}
\end{eqnarray}
Using Eqs. \eqref{eq:noise_p_cri_unknown} and \eqref{eq:noise_p_known_critical_thermal_a_b}, we have $\chi_{\bcancel{p}}^{max}(\rho^{\beta_1 \beta_2}) > 2$ when
\begin{eqnarray}
p < \frac{\sqrt{1+\tanh^2 2r}(\sqrt{1+\tanh^2 2r}-1)}{1+\tanh^2 2r - \tanh \frac{\beta_1}{2} \tanh \frac{\beta_2}{2}}.
\label{eq:noise_p_cri_unknown_ther}
\end{eqnarray}
Clearly, the parameter space, in which violation of Bell inequality occurs in the $p$-unknown scenario, is smaller compared to the case when $p$ is known. For example, if $\beta_1,\beta_2$, and $r$ are taken to be $3,5$ and $1.25$ respectively, we obtain violation for $p < 0.633$ when $p$ is known, and for $p < 0.526$ with $p$ being unknown.
The distinction becomes more pronounced in the low temperature limit of the noise. When $\beta_1,\beta_2 \rightarrow \infty$, knowledge of $p$ guarantees that the state keeps violating Bell inequality for all values of $p < 1$. On the contrary, when the knowledge of the value of $p$ is absent, the state in the above limit violates Bell inequality only when $ p < \frac{\sqrt{1+\tanh^2 2r}(\sqrt{1+\tanh^2 2r}-1)}{\tanh^ 2r}$. For the EPR state, this bound reduces to $p< 2 - \sqrt{2}$. On the other hand, in the high temperature limit $(\beta_1,\beta_2 \rightarrow 0)$, the violation becomes insensitive to the knowledge of $p$, and in both the cases, the state violates Bell inequality for $p < 1 - 1/{\sqrt{1+\tanh^2 2r}}$, which reduces to $p < 1 - 1/{\sqrt{2}}$ for the EPR state. This is reminiscent of the result involving continuous variable Werner state in \cite{cv_werner}. Note here that although the properties of states with known $p$ have been studied before, the situation when $p$ is unknown, although very relevant has hardly been investigated before.


\begin{figure}[h]
\includegraphics[width=0.6\linewidth]{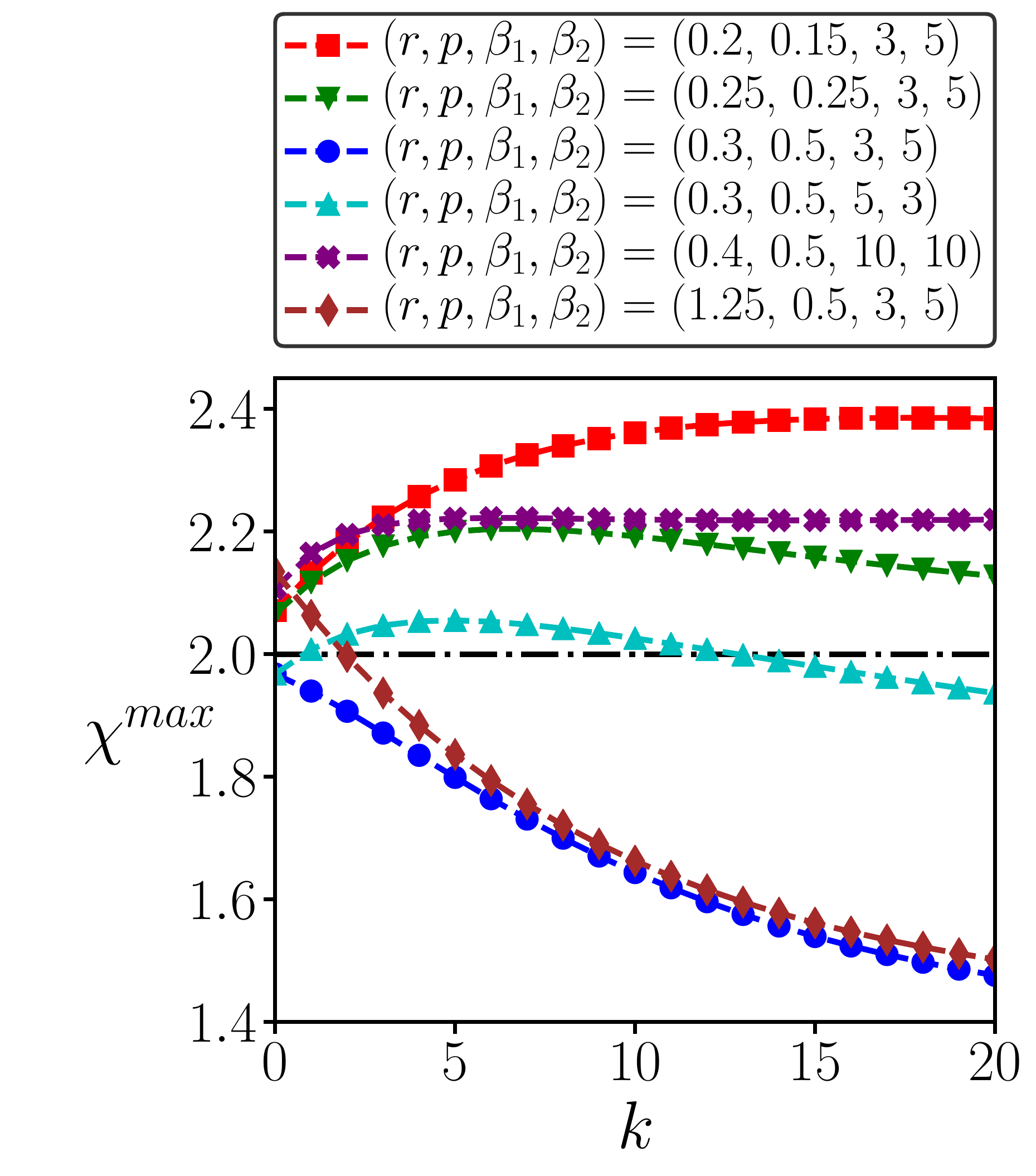}
\caption{Variation of the maximal violation of Bell inequality for photon-added TMSV states mixed with local thermal noise against number of added photons $k$. We choose different values of ($r,p, \beta_1,\beta_2$) to make the observation more prominent. Both axes are dimensionless.}
\label{fig:noise_ther_add}
\end{figure}

Let us first analyze how the photon addition and subtraction process effect the violation of Bell inequality when $p$ is known. 
When $k$ photons are added to these states, $A_{+k}$ in Eq. \eqref{eq:noise_A,B_add_k} becomes 
\begin{eqnarray}
(1-p) + p\tanh^{k+1} \frac{\beta_1}{2} \tanh \frac{\beta_2}{2},
\end{eqnarray}
 while $B_{+k}$ is same as given in Eq. \eqref{eq:noise_A,B_add_k}. 
When an even number of photons are subtracted, $k$, we have $A_{-k} = A_{+k}$ and $B_{-k} = B_{+k}$, leading to the same $\chi^{max}$ for photon-added and -subtracted states. However, for the case of subtracting an odd number of photons, $B_{-k}$ remains same but $A_{-k} =(1-p) - p \tanh^{k+1} \frac{\beta_1}{2} \lbrace \tanh \frac{\beta_2}{2} - (1-e^{-\beta_2})\rbrace$. The above expressions clearly indicate that in the presence of local thermal noise, addition and subtraction of photons are equivalent. However, when odd number of photons are involved, addition performs better than subtraction in terms of the maximal violation. Therefore, we restrict ourselves to single mode operations only involving photon addition. Nevertheless, similar analysis can also be caried out for photon subtraction. 

Note that when photons are added to 
$\rho^{\beta_1\beta_2}$, $B_k$ $\big( \sim \mathcal{K}_{(k,0)} \big)$ shows overall enhancement,  while $A_k$ monotonically decreases following the decrement of the term $\tanh^{k+1} \frac{\beta_1}{2}$ with $k$ in its expression. Therefore, the maximal value of Bell expression, $\chi^{max}_p = 2\sqrt{A_k^2 + B_k^2}$, is not guaranteed to increase after adding photons, and is determined by the competing enhancement and decrement of $B_k$ and $A_K$ respectively. In Fig. \ref{fig:noise_ther_add}, we plot $\chi^{max}_p$ for various values of system parameters for known $p$, which encapsulates the following patterns:

\begin{enumerate}[label=\textbf{\arabic*.}]
\item For low values of $p$, the noisy state is essentially close to the TMSV state, and therefore, we get enhancement in the Bell expression on addition of photons (see the curve with $(r, p, \beta_1, \beta_2) = (0.2, 0.15, 3, 5)$ in Fig. \ref{fig:noise_ther_add}). 

\item For low values of the squeezing parameter, $r$, the overall gain $\mathcal{G}(\ket{\psi_r^{(k,0)}})$ on addition of photons to the TMSV state is large (see Table. \ref{table:gain1}). Now, in the low to intermediate temperature regime, when $p$ is small, the increase of $\mathcal{K}_{(k,0)}$ (due to the high gain) dominates, and therefore the overall violation `increases', and ultimately saturates for high values of $k$.

\item There exists regions in the parameter space, where we can have `activation' of violation of Bell inequality, i.e., the state which is originally non-Bell violating, violates local realism after adding $k$ photons. See the plot with $(r,p,\beta_1,\beta_2) = (0.3,0.5,5,3)$ in Fig. \ref{fig:noise_ther_add}. Also note that, in this situation, the value of the  Bell expression initially increases with the number of added photons, but it starts decreasing after sometime, as the decrement of the term $\tanh^{k+1}\frac{\beta_1}{2}$ in $A_k$ becomes dominating.

\item For high values of the squeezing parameter, $r$, the value of  $\mathcal{K}_{(k,0)}$ does not change substantially. It is reflected in the low gain percentages for highly squeezed TMSV states in Table. \ref{table:gain1}. Therefore, the Bell expression decreases monotonically with $k$. Similarly, for high values of $p$, the Bell expression can decrease, as $\tanh^{k+1}\frac{\beta_1}{2}$ in $A_k$ dominates.
 \end{enumerate}

 For the case of unknown $p$, we observe qualitatively the same features as in the scenario for which $p$ is known but with reduced values of the maximal violation.

\subsubsection{Local Gaussian noise}

We now admix the TMSV state with local Gaussian noise, denoted by $\rho^{\sigma_1\sigma_2}$ having coefficients 
\begin{eqnarray}
\mu_n &=& \frac{2}{1+\vartheta_3(0,e^{-\sigma_1^{-2}})}e^{-n^2/\sigma_1^2}, \nonumber \\
\nu_n &=& \frac{2}{1+\vartheta_3(0,e^{-\sigma_2^{-2}})}e^{-n^2/\sigma_2^2}
\end{eqnarray}
where $\sigma_1$ and $\sigma_2$ are the relevant noise parameters, and $\vartheta_n$ denotes the Jacobi theta function of order $n$ \cite{theta}. In this case, 
\small
\begin{eqnarray}
 A = (1-p) + p \frac{1+\vartheta_4(0,e^{-\sigma_1^{-2}})}{1+\vartheta_3(0,e^{-\sigma_1^{-2}})} \times \frac{1+\vartheta_4(0,e^{-\sigma_2^{-2}})}{1+\vartheta_3(0,e^{-\sigma_2^{-2}})},
\end{eqnarray}\normalsize
and $B$ remains the same as in Eq. \eqref{eq:noise_A,B}.

Like in the case of thermal noise, for a given $p$, the maximal violation of Bell inequality  using Eq. \eqref{eq:noise_bv_p_known}, takes the form as 
\begin{widetext}
\begin{eqnarray}
\chi^{max}_p(\rho^{\sigma_1 \sigma_2}) = 2\sqrt{(1-p)^2 \tanh^2 2r + \Big\lbrace(1-p) + p\times\frac{1+\vartheta_4(0,e^{-\sigma_1^{-2}})}{1+\vartheta_3(0,e^{-\sigma_1^{-2}})}\times\frac{1+\vartheta_4(0,e^{-\sigma_2^{-2}})}{1+\vartheta_3(0,e^{-\sigma_2^{-2}})}\Big\rbrace^2}.
\label{eq:noise_gaussian_BV_gen}
\end{eqnarray}
\end{widetext}
In  case of local Gaussian noise, when $p$ is unknown, the violation, following Eq. \eqref{eq:noise_bv_p_unknown} reads as
\begin{eqnarray}
\chi^{max}_{\bcancel{p}}(&\rho^{\sigma_1 \sigma_2}) = \frac{2}{\sqrt{1+\tanh^2 2r}}\Big[(1-p) (1+\tanh^2 2r) \nonumber \\ 
&+ p  \frac{1+\vartheta_4(0,e^{-\sigma_1^{-2}})}{1+\vartheta_3(0,e^{-\sigma_1^{-2}})}\times \frac{1+\vartheta_4(0,e^{-\sigma_2^{-2}})}{1+\vartheta_3(0,e^{-\sigma_2^{-2}})}\Big].
\label{eq:noise_Gaussian_BV_p_unknown}
\end{eqnarray}
In such a situation, $\chi^{max}_{\bcancel{p}}(\rho^{\sigma_1 \sigma_2}) > 2$ holds for
\begin{eqnarray}
p < \frac{\sqrt{1+\tanh^2 2r}(\sqrt{1+\tanh^2 2r}-1)}{1+\tanh^2 2r - \frac{1+\vartheta_4(0,e^{-\sigma_1^{-2}})}{1+\vartheta_3(0,e^{-\sigma_1^{-2}})} \frac{1+\vartheta_4(0,e^{-\sigma_2^{-2}})}{1+\vartheta_3(0,e^{-\sigma_2^{-2}})}}.
\label{eq:noise_p_cri_unknown_gauss}
\end{eqnarray}


\subsection*{Classically correlated noise}
Instead of uncorrelated noise considered in Eq. \eqref{eq:noise_local}, we now move to classically correlated local noise model, and examine Bell inequality violations for these states. Such a state can be represented as
\begin{eqnarray}
\tilde{\rho} = (1-p)|\psi_r\rangle \langle \psi_r| + p\sum_{n=0}^{\infty} C_n |n,n \rangle \langle n,n|. 
\label{eq:noise1}
\end{eqnarray}
The correlation functions for $\tilde{\rho}$, from Eq. \eqref{eq:correlation_fn}, in terms of the pseudospin operators are given by
\begin{eqnarray}
E(\theta_a, \theta_b) =  \cos \theta_a \cos \theta_b + \tilde{\mathcal{K}}_{(0,0)}\sin \theta_a \sin \theta_b,
\end{eqnarray}
with $\tilde{\mathcal{K}}=(1-p)\mathcal{K}_{(0,0)}=(1-p)\tanh 2r$. 
For known $p$,  ${\chi}^{max}_{{p}}(\tilde{\rho}) = 2\sqrt{1 + \tilde{\mathcal{K}}^2}$. In this situation, it is easy to see that $\chi^{max}_p > 2$ for any values of $p < 1$, and for any finite values of the squeezing parameter, $r$.

On the other hand, the maximal Bell expression takes the form
\small
\begin{eqnarray}
{\chi}^{max}_{\bcancel{p}}(\tilde{\rho}) = 2\Big( \frac{1+\mathcal{K}_{(0,0)}\tilde{\mathcal{K}}_{(0,0)}}{\sqrt{1 + \mathcal{K}^2_{(0,0)}}}\Big)
 = 2 \Big( \frac{1+(1-p)\tanh^2 2r}{\sqrt{1 + \tanh^2 2r}}\Big). \nonumber \\
\label{eq:noise2}
\end{eqnarray} \normalsize
with unknown $p$. Under this assumption about the uncertainty in the error estimation/detection, we observe criticalities in the values of $r$ and $p$, beyond which the system ceases to violate the Bell inequality based on pseudospin operators. 
For $r\rightarrow 0$, we find that ${\chi}^{max}_{\bcancel{p}}(\tilde{\rho}) \leq 2$ for $p \geq 1/2$. Therefore, for $p < 1/2$, the noisy state violates the Bell inequality for any finite squeezing, even when the value of $p$ is not known. However, if $p > 1/2$, the state given in Eq. \eqref{eq:noise1} starts violating the Bell inequality only when
\begin{eqnarray}
2r \geq \tanh^{-1}\frac{\sqrt{1-2(1-p)}}{1-p}.
\label{eq:critical_r_diag_noise}
\end{eqnarray}   
Note that if $p \geq 2-\sqrt{2}$, even the EPR state, i.e., the TMSV state with $r\rightarrow\infty$, does not violate a Bell inequality in this setting. So we get a criticality in the squeezing parameter, $r$, given by the above equation, when $1/2 \leq p \leq 2-\sqrt{2}$.

When we add or subtract photons to a single mode of the state given in Eq. \eqref{eq:noise1},  we have $\tilde{\mathcal{K}}_{(\pm k,0)} = (1-p)\mathcal{K}_{(k,0)} = (1-p)\mathcal{K}_{(-k,0)} = \tilde{\mathcal{K}}_{(k,0)} $. Now, if $p$ is known, the maximal violation of Bell inequality simply reads 
\begin{eqnarray}
\chi^{max}_{p} (\tilde{\rho}_{\pm k}) = 2\sqrt{1 + \tilde{\mathcal{K}}_{(k,0)}^2}.
\end{eqnarray}
From the above expression, it is clear that photon addition (subtraction) always leads to an overall enhancement in the violation of Bell inequality, which will be dictated by the change in $\mathcal{K}_{(k, 0)}$ with respect to $k$.
In the absence of any knowledge about $p$, the violation is given by
\begin{eqnarray}
\chi^{max}_{\bcancel p} (\tilde{\rho}_{\pm k}) = 2\Big( \frac{1+\mathcal{K}_{(k,0)}\tilde{\mathcal{K}}_{(k,0)}}{\sqrt{1 + \mathcal{K}^2_{(k,0)}}}\Big).
\label{eq:noise3}
\end{eqnarray}
Again, upon addition (subtraction) of photons in one mode, the above expression can be increased. Specifically, for $1/2 \leq p \leq 2 - \sqrt{2}$ and $r < \tanh^{-1}\left(\sqrt{1 - 2(1-p)}/(1-p)\right)$, the violation can be \emph{activated} via photon addition or subtraction in a single mode.

Interestingly, note that, for any noise with the same structure ($\sum C_n |n,n \rangle \langle n,n|$) as given in Eq. \eqref{eq:noise1}, the Bell expressions for known or unknown values of $p$ do not depend of the values of $C_n$.

%

\subsection{Faulty Twin Beam Generator}
\label{sec:fault}
Upto now, we consider the scenario where the state is affected by noise. There  can be a situation where the twin beam generator is typically imperfect, and due to internal imperfection and losses, it may end up in generating TMSV states with less squeezing than it is ought to. 
To put things in a quantitative perspective, we assume that a twin beam generator which is labeled to produce states with squeezing $r$, does so with an unknown $r'$. Off course, $r' < r$. The correlators are calculated via measurements performed with pseudospin operators oriented in the optimal direction for the TMSV state with squeezing parameter $r$. The maximal violation obtained in such a situation, following Eqs. \eqref{eq:optimal_setting} and \eqref{eq:max_bv}, is given by
\begin{eqnarray}
\chi^{max}_{r} (\ket{\psi^{r'}}) &=& 2 \times (\cos \theta +\tanh 2r' \sin \theta), \nonumber \\
&=& 2 \times \frac{1+\tanh 2r'\tanh 2r}{\sqrt{1 + \tanh^2 2r}}
\label{eq:fault1}
\end{eqnarray}
We have $\chi^{max}_{r} (\ket{\psi^{r'}}) \leq 2$ when
\begin{eqnarray}
\tanh 2r' \leq \frac{1}{\tanh 2r} \times \big(\sqrt{1 + \tanh^2 2r}-1 \big).
\label{eq:fault_r_critical}
\end{eqnarray}
The equality holds when $r' = r_c$, where $r_c$ is the critical value of $r'$ for a given $r$, below which the state fails to show any violation. For the EPR state, the critical value of $r'$ saturates to a finite value $r_c^{\infty} = \frac{1}{2} \times \arctanh (\sqrt{2}-1) \approx 0.22$. Although $r_c$  might seem to be a small value even for the EPR state, for experimentally relevant squeezing parameters, $r_c$ is comparable to $r$. For example, $r_c$ for  $r=0.75$ approximately reads $0.203172$. We want to analyze the effects of adding or subtracting photons from a single mode when $r'$ falls below $r_c$, i.e., when the state does not violate Bell inequality based on pseudospin operators.

Let us check whether the range of squeezing parameter which shows non violation can be changed if one adds (subtracts) photons even in a single mode. In this case, for single mode operations, photon addition remains equivalent to subtraction,  since the fault in the generator just reduces the squeezing parameter of the TMSV state, and thus equivalence argument  goes through.  So we add $k$ photons to the first mode of the obtained squeezed state with squeezing $r'$. As before, we use the optimal measurement settings for $k$ photon-added TMSV state with squeezing parameter $r$. Note that $r$ is the labeled value of squeezing that the twin beam generator is intended to produce. The maximal violation, so obtained in this scenario, is given by 
\begin{eqnarray}
\chi^{max}_r (\ket{\psi_{r'}^{(\pm k,0)}}) = 2 \times \frac{1+\mathcal{K}_{(k,0)}^{r'} \mathcal{K}_{(k,0)}}{\sqrt{1 + \mathcal{K}_{(k,0)}^2}},  
\end{eqnarray}
where $\mathcal{K}_{(k,0)}$ is in Eq. \eqref{eq:K_single_mode}, and is calculated for the TMSV state with squeezing parameter $r$, while $\mathcal{K}_{(k,0)}^{r'}$ represents the same thing for the state $|\psi^{(k,0)}_{r'}\rangle$. We consider some representative states for which $r'$ falls below the critical value as given in Eq. \eqref{eq:fault_r_critical}.
We show that it is possible to \emph{activate} the violation for these states by using single mode operations. See Fig. \ref{fig:activation_faulty} (a).

\begin{figure}[ht]
\includegraphics[width=\linewidth]{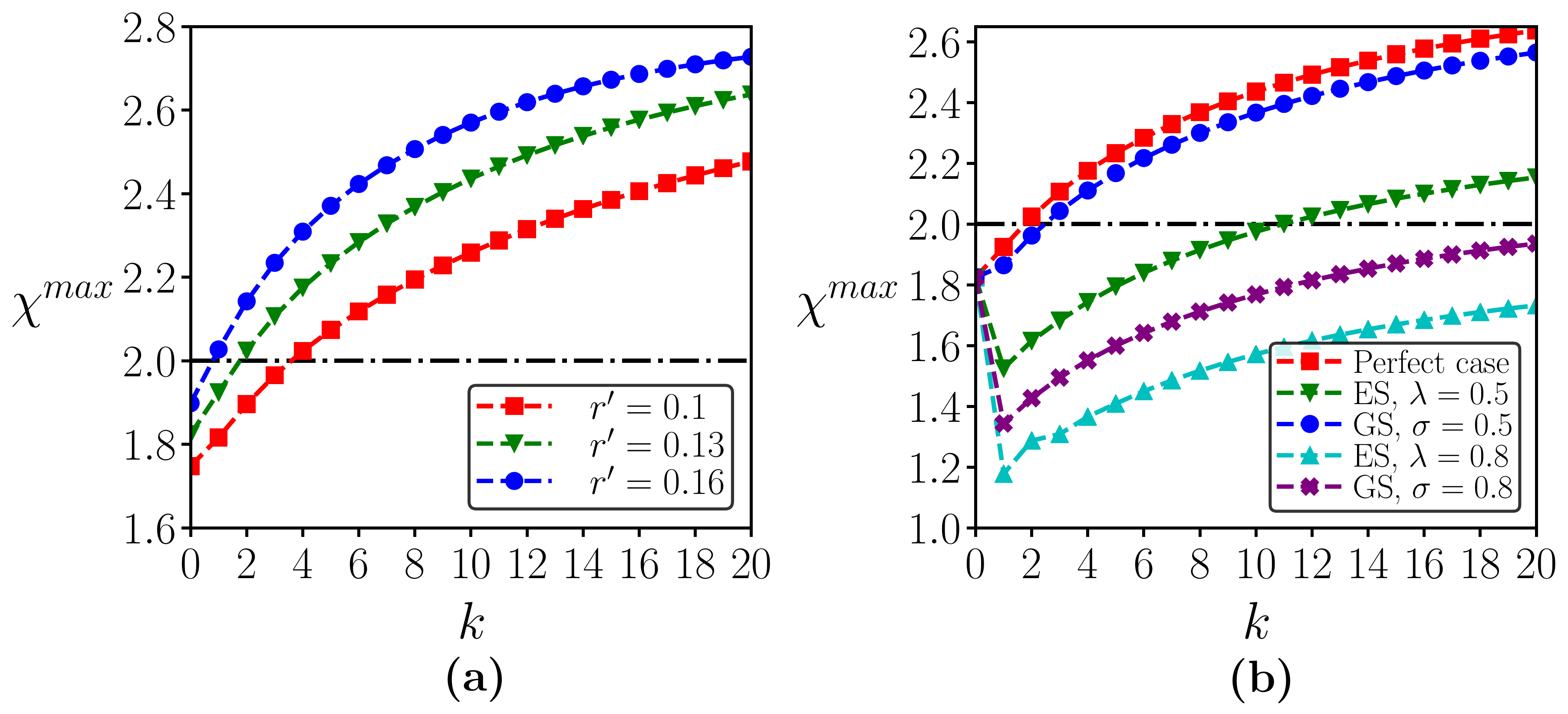}
\caption{Activation of the violation of pseudospin operator-based Bell inequality. We choose $r=0.75$, such that the corresponding critical value of $r'$ is $r_c = 0.203172$. (a) Activation via deterministic addition (subtraction) of photons in single mode for three values of $r' < r_c$. (b) The dual effects of imperfect addition (subtraction) of photons in single mode and faulty TMSV state for $r' = 0.13$, and for different values of $\lambda$ and $\sigma$ (see discussions in Sec. \ref{sec:dark}). Both the axes are dimensionless.}
\label{fig:activation_faulty}
\end{figure}

In the entire analysis, we have assumed that the addition or subtraction process is error free. In the next subsection, the same situations will be re-examined, when the photon addition or subtraction process is itself imperfect.

%
%

\subsection{Imperfections in photon addition and subtraction mechanism}
\label{sec:dark}

The indeterminacy in addition (subtraction) of photons can be attributed due to variety of reasons, like dark counts \cite{noise_dark,nd2} of the detector etc., and hence can lead to decrement in violation of local realism. In this subsection,  we consider two distinct models of imperfections in the added or subtracted number of photons. Firstly, for a given $k$ number of added (subtracted) photons, we assume that the state is to be mixed with  $k-1, k-2,...k-m$ ($m\leq k$) number of photon-added (-subtracted) states with  probabilities which follows the exponential suppression (ES). Hence such that the effective  state becomes
\begin{eqnarray}
\bar{\rho}_{\pm k} = \sum_{i = 0}^m p_i  \rho_{\pm |k-i|}.
\end{eqnarray} 
Here, $m$ is the cutoff on the maximal discrepancy in the photon number during the addition (subtraction) procedure, $\rho_{\pm l}$ represents a state with $l$ number of added (subtracted) photons, and
$p_i$s are the mixing probabilities, which decrease according to exponential law. Second scenario considered in this paper where
 the probabilities are Gaussian. 
Specifically, the exponential probabilities, for a given $m$, are given by
\begin{eqnarray}
p_i = \frac{e^{-i/\lambda}}{\sum_{i=0}^m e^{-i/\lambda}},
\end{eqnarray}
whereas  for Gaussian suppression (GS), the probabilities takes the form as 
\begin{eqnarray}
p_i = \frac{e^{-i^2/\sigma^2}}{\sum_{i=0}^{m} e^{-i^2/\sigma^2}}.
\end{eqnarray}
Here, $\lambda$ and $\sigma$ give the measures of dispersion for these imperfect additions (subtractions).  


\subsubsection{Noisy states}
In Sec. \ref{sec:noise}, we have discussed the effects of noise on the violation of Bell inequality for the TMSV states, and the role of photon addition and subtraction to improve the situation. Specifically, we have discussed the cases of local noises (thermal and Gaussian), and a classically correlated noise. In this subsection, we  study the effects of faulty addition (subtraction) of photons on the Bell expression, when the noise probability, $p$, is known. 

When the photon addition scheme on TMSV states with local noise suffers exponential suppression, and when we know the value of $p$, the violation of Bell inequality is given by  
\begin{eqnarray}
\chi_{ES}^{max} &=& \frac{2(\sum_{n=0}^{m} e^{-n/\lambda})^{-1}}{\sqrt{A_{k}^2 + B_k^2}}  \Big [ \sum_{i=0,2,4,... \leq m} e^{-i/\lambda}\big(A_{k}A_{k-i} \nonumber \\ &+& B_{k}B_{k-i}\big)  
-  \sum_{j=1,3,5,... \leq m} e^{-j/\lambda} A_{k}A_{k-j} \Big ]. 
\label{eq:noise_acti_exp}
\end{eqnarray}  
The corresponding violation for GS reads as
\begin{eqnarray}
\chi_{GS}^{max} &=& \frac{2(\sum_{n=0}^{m} e^{-n^2/\sigma^2})^{-1}}{\sqrt{A_{k}^2 + B_k^2}}  \Big [ \sum_{i=0,2,4,... \leq m} e^{-i^2/\sigma^2}\big(A_{k}A_{k-i} \nonumber \\ &+& B_{k}B_{k-i}\big)  
-  \sum_{j=1,3,5,... \leq m} e^{-j^2/\sigma^2} A_{k}A_{k-j} \Big ]. 
\label{eq:noise_acti_g}
\end{eqnarray}  
Here, $A_k$ and $B_k$  are given in Eqs. \eqref{eq:noise_A,B_add_k}. In this imperfect addition scenario, for low values of noise parameters, the enhancement in the maximal violation persists, with lower values compared to the perfect addition scenario (see Fig. \ref{fig:noise_ther_add}). The domain of activation also naturally shrinks in this imperfect case.
 \begin{figure}[h]
\includegraphics[width=\linewidth]{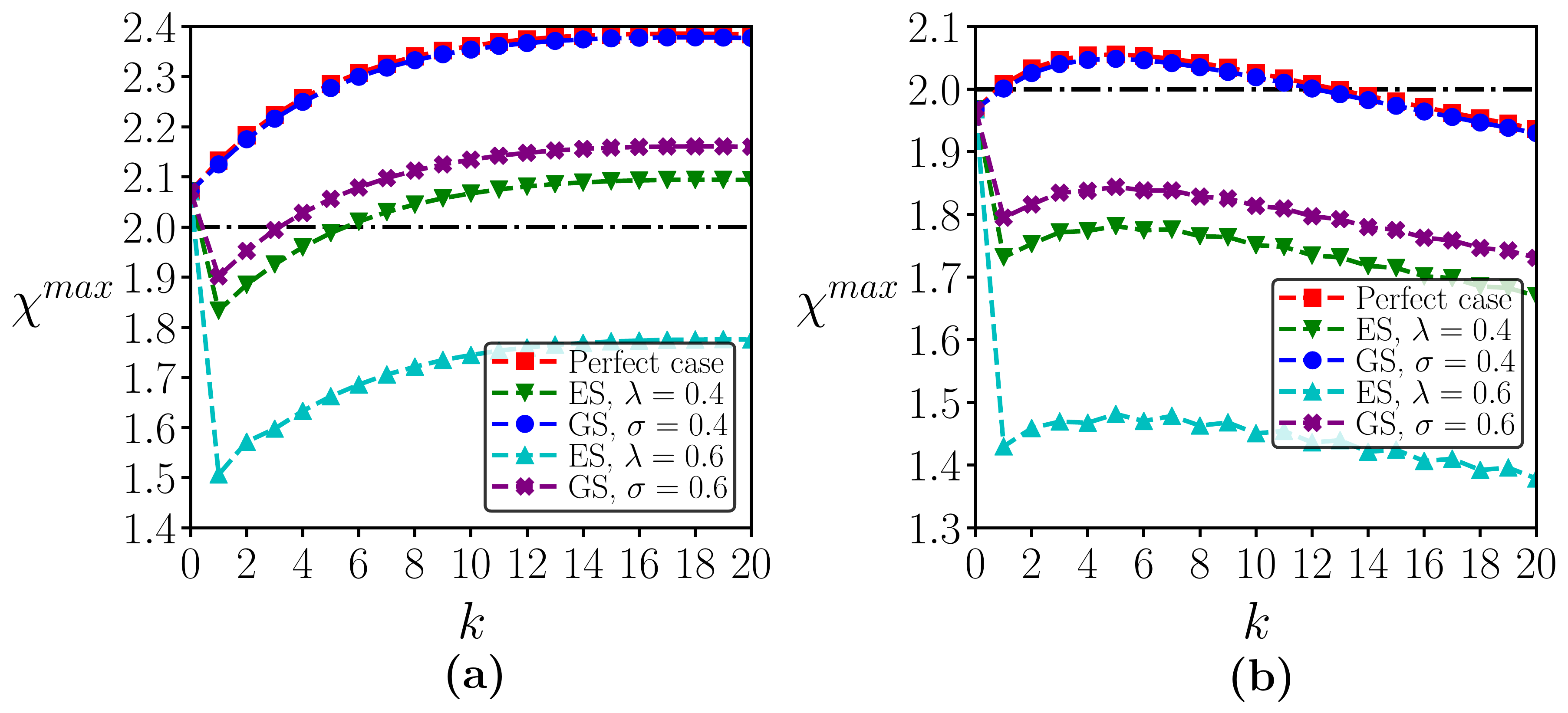}
\caption{Effects of imperfect photon addition process on TMSV states with local thermal noise, (a) with $(r,p,\beta_1,\beta_2)=(0.2,0.15,3,5)$, and (b) with $(r,p,\beta_1,\beta_2)=(0.3,0.5,5,3)$.  See Fig. \ref{fig:noise_ther_add} for the perfect photon addition case for these choices of system-noise parameters. We chose these two values to highlight the consequences of imperfect photon addition on situations of enhancement (a) and activation (b) of $\chi^{max}$. Both the axes are dimensionless.}
\label{fig:activation_noisy}
\end{figure}

\subsubsection{Faulty twin beam generator}
\label{sec:faulty_faulty}

We now study the response on violation of Bell inequality under coupled imperfect scenario. In particular, along with imperfect photon addition, the twin beam generator produce the TMSV state with $r'$, instead of $r$. In case of ES,
\begin{eqnarray}
\chi_{ES}^{max} &=& \frac{2(\sum_{n=0}^{m} e^{-n/\lambda})^{-1}}{\sqrt{1 + \mathcal{K}_{(k,0)}^2}}   \Big [ \sum_{i=0,2,4,... \leq m} e^{-i/\lambda}\big(1 \nonumber \\ 
 &+&\mathcal{K}_{(k-i,0)}^{r'} \mathcal{K}_{(k,0)}\big)  
 -  \sum_{j=1,3,5,... \leq m} e^{-j/\lambda} \Big], 
\label{eq:fault_acti_exp}
\end{eqnarray}  
while for GS, we get
\begin{eqnarray}
\chi_{GS}^{max} &=& \frac{2(\sum_{n=0}^{m} e^{-n^2/\sigma^2})^{-1}}{\sqrt{1 + \mathcal{K}_{(k,0)}^2}}   \Big [ \sum_{i=0,2,4,... \leq m} e^{-i^2/\sigma^2}\big(1 \nonumber \\ 
 &+&\mathcal{K}_{(k-i,0)}^{r'} \mathcal{K}_{(k,0)}\big)  
 -  \sum_{j=1,3,5,... \leq m} e^{-j^2/\sigma^2} \Big]. 
\label{eq:fault_acti_gauss}
\end{eqnarray}  
To take one concrete example, we restrict $m$ to be equal to $k$, and choose different values of $\lambda$ and $\sigma$, and examine the consequence of faulty photon addition procedure on the Bell expression (see Fig. \ref{fig:activation_faulty} (b)). 
We observe that for low values of dispersions ($\lambda$ and $\sigma$), the Bell expression, which initially does not violate, 
always increase with varying number of added photons, $k$, leading to activation of the violation. However, if the dispersions are large, in both exponential and Gaussian cases,
the Bell expression initially shows a decrement in its value, and can finally be enhanced or activated after adding sufficiently high number of photons. There can also exist scenarios, where this activation is not possible at all, even after adding a large number of photons (see Fig. \ref{fig:activation_faulty} (b)).


\section{Conclusion} \label{sec:conclu}
Violation of Bell inequalities by quantum systems establishes the existence of correlations beyond the classical ones. For 
finite-dimensional quantum systems, violation of Bell inequalities  have been studied more thoroughly in comparison to the same for  continuous-variable (infinite dimensional) systems.
 In the field of continuous-variable systems, there was a long outstanding debate, started by John Bell, as to whether states with positive Wigner functions would violate a Bell inequality. It was resolved conclusively by constructing Bell expressions using parity operators, later using pseudospin operators, and demonstrating violation for certain entangled states with positive Wigner function. In this paper, we used pseudospin operators to examine the violation of Bell inequality for photon-added and -subtracted two-mode squeezed vacuum (TMSV) states, where addition as well as subtraction is performed either in a single mode or in both the modes. We found that unlike entanglement, the amount of violation of pseudospin operator-based Bell inequality by photon-subtracted state can be higher than that of the photon-added ones. 
  

We have further studied the effects of local noise (specifically, local thermal and local Gaussian noise) on the maximal violation of Bell inequality for the TMSV states, and computed the parameter ranges for which the noisy TMSV state abstains from violating the considered Bell inequality. We demonstrated that under such circumstances, single-mode operations like photon addition, can \emph{activate}  violation. We repeated the same drill of investigations with the goal of activation in the case of a faulty 
twin-beam generator for generating TMSV states, and imperfections in photon addition or subtraction process. We reported here that in both the scenarios, the answer is affirmative, i.e., the activation is possible, thereby transforming non-violating states to violating ones.

%
%

\appendix

\section{Maximization of Bell expression}
Let us now discuss the method for obtaining $\chi^{max}$, by performing maximization over the settings, i.e., $\theta_a,\theta_b,\theta_a^{'},\theta_b^{'}$.
The correlation matrix, $T_{ij} = \langle S^i_{q_1} \otimes S^j_{q_2} \rangle$, where $i,j = x,y,z$ and $q_1$ and $q_2$ are chosen appropriately depending on the structure of the state. $S^x_q$ and $S^y_q$ are simply given by $S^+_q + S^-_q$ and $-i(S^+_q - S^-_q)$ respectively, where $S^+_q$ and $S^-_q$ are given in Eq. \eqref{eq:gen_pseudo_spin}. The $T$ (correlation) matrix for photon-added (subtracted) TMSV states
can be expressed as, $T = \text{ diag}(\mathcal{K}, -\mathcal{K},1)$. The two highest eigenvalues of the matrix $T^\dagger T$ are $1$ and $\mathcal{K}^2$ respectively. Now, following the argument as given in \cite{horo_appendix}, the maximal violation of Bell inequality is given by $2 \times \sqrt{\mathcal M(T^\dagger T)}$, where $\mathcal M(T^\dagger T)$ represents the sum of the two largest eigenvalues of $T^\dagger T$. The same in this case reads
 $2\sqrt{1+\mathcal{K}^2}$. This completes the proof of the expressions in Eqs. \eqref{eq:max_bv} and \eqref{eq:noise_bv_p_known}.

For the two-qubit system, the conditions derived in \cite{horo_appendix} provides a necessary and sufficient condition for violations of Bell inequality. However, the conditions of violation of Bell inequality obtained in this manuscript using pseudospin operators are only sufficient but not necessary. This is so because for the two-qubit states, the Pauli spin operators form a basis for all operators in that space. Thus, the optimization involves the maximization of the entire spectrum of dichotomic operators in the two-qubit space. However, the pseudospin operators do not form a basis for all dichotomic operators in the space of two mode continuous variable states. Therefore, the maximization involved here, only gives the maximal violation of Bell inequality in the restricted subspace of pseudospin operators making the conditions of violation sufficient but not necessary.

\label{appendix:A}

\end{document}